\documentclass{article}
\usepackage[margin=1in]{geometry}
\usepackage{hyperref}
\usepackage{microtype}
\usepackage{amsthm}
\usepackage{mathrsfs}
\usepackage{tabularx}
\usepackage{xfrac}
\usepackage{enumitem}
\usepackage{tikz}
\usepackage{circuitikz}
\usetikzlibrary{calc}
\ctikzset{logic ports=ieee, logic ports/scale=0.8}
\usepackage[export]{adjustbox}

\usepackage{graphicx}
\graphicspath{{fig/}}
\newcommand\fig[2][]{\includegraphics[scale=0.394,#1]{#2}}

\usepackage[backend=biber,style=alphabetic]{biblatex}
\newtheorem{theorem}{Theorem}[section]
\newtheorem{lemma}[theorem]{Lemma}
\newtheorem{corollary}[theorem]{Corollary}
\newtheorem{proposition}[theorem]{Proposition}
\newtheorem{claim}[theorem]{Claim}

\theoremstyle{definition}
\newtheorem{definition}[theorem]{Definition}
\newtheorem{problem}[theorem]{Problem}

\def\C{\mathscr C}
\def\U{\mathscr U}
\def\concat{\mathbin{+\!+}}

\usepackage{caption}
\newcommand*{\defn}[1]{\textbf{\textit{\boldmath{#1}}}}

\begin{document}

\title{Complexity of Planar Graph Orientation Consistency, Promise-Inference, and Uniqueness, \\
  with Applications to Minesweeper Variants}

\author{%
  MIT Hardness Group%
    \thanks{Artificial first author to highlight that the other authors (in
      alphabetical order) worked as an equal group. Please include all
      authors (including this one) in your bibliography, and refer to the
      authors as ``MIT Hardness Group'' (without ``et al.'').}
\and
  Della Hendrickson%
    \thanks{MIT Computer Science and Artificial Intelligence Laboratory,
      32 Vassar St., Cambridge, MA 02139, USA, \protect\url{{della,tockman}@mit.edu}}
\and
  Andy Tockman\footnotemark[2]
}

\date{}

\maketitle

\begin{abstract}
    We study three problems related to the computational complexity
    of the popular game Minesweeper.
    The first is consistency:
    given a set of clues,
    is there any arrangement of mines that satisfies it?
    This problem has been known to be NP-complete since 2000 \cite{cons},
    but our framework proves it as a side effect.
    The second is inference:
    given a set of clues,
    is there any cell that the player can prove is safe?
    The coNP-completeness of this problem has been in the literature since 2011 \cite{inf},
    but we discovered a flaw that we believe is present in all published results,
    and we provide a fixed proof.
    Finally,
    the third is solvability:
    given the full state of a Minesweeper game,
    can the player win the game by safely clicking all non-mine cells?
    This problem has not yet been studied,
    and we prove that it is coNP-complete.
\end{abstract}

\section{Introduction}

The puzzle-based video game Minesweeper
was popularized by its inclusion in the default installation of Windows 3.1 in 1992,
and to this day is one of the most widely recognizable computer games.
Though the premise is very simple,
it has spawned an endless stream of spinoffs,
and a number of communities dedicated to challenges
such as completing a randomly generated game as quickly as possible.

From a computational complexity perspective,
Minesweeper is interesting
in that there are multiple natural decision problems to study.
The simplest is \emph{consistency}
(`given some clues, is there a possible arrangement of mines?'),
which was proved NP-complete in 2000 \cite{cons}.

But consistency doesn't capture the essence of Minesweeper as a game,
where the player has some partial information and tries
to find a cell that they can click on safely, leading to more information.
This suggests the \emph{inference} problem
(`given some clues, is there a cell that is provably safe?'),
which was shown coNP-complete in 2011 \cite{inf}.
Unfortunately,
this proof of coNP-hardness of Minesweeper inference
is incorrect,
and Thieme and Basten's proof \cite{indeed},
which is designed to use very small gadgets,
suffers from the same issue.

While inference asks about a single `turn' of Minesweeper,
it is also natural to ask a related question about the entire game.
We introduce the \emph{solvability} decision problem,
in which we are given both the current state of a Minesweeper game
and also the secret arrangement of mines,
and asked whether the player can make a sequence of safe clicks
to solve the game.

In this paper, we develop a framework based on graph orientation
to prove coNP-completeness of Minesweeper inference and solvability.
As a side effect, the same gadgets also suffice to prove NP-hardness
of Minesweeper consistency.
Of particular note, we prove that Minesweeper solvability is coNP-complete
even `after a single click':
from a nearly empty initial state with only one cell revealed.

Specifically, we define three graph orientation
decision problems (consistency, promise-inference, and uniqueness)
related to the Minesweeper problems,
and show that each is hard with a particular set of simple abstract gadgets.
It follows that finding well-behaved constructions in Minesweeper
which behave like those gadgets
is enough to prove hardness for all three Minesweeper problems.

In Section~\ref{sec:prior},
we define each decision problem carefully,
summarize the previously known results,
and explain the flaw in the existing reduction for inference.
In Section~\ref{sec:framework},
we develop a framework of gadgets
that makes it easy to prove hardness results
for all three decision problems.
% In Section~\ref{sec:min},
% we expand upon the framework
% by explicitly listing several minimal sets of gadgets
% which are enough for hardness.
Finally,
in Section~\ref{sec:variants} we apply the framework
to Minesweeper
and to many variants of Minesweeper
from the video game \emph{14 Minesweeper Variants}.
For the most part, each application consists entirely
of constructions of the relevant gadgets
in the Minesweeper variant under consideration.
% TODO Finally,
% in Section~\ref{sec:summary}
% we summarize known hardness for all combinations of variants,
% as well as the cases that remain open.

\section{Prior work and definitions} \label{sec:prior}

\subsection{Consistency}

Research on the computational complexity of Minesweeper
began when Sadie Kaye \cite{cons}
posed the Minesweeper consistency problem.
Informally,
this problem asks whether a partially completed Minesweeper board
has a legal arrangement of mines.
We provide a formal definition here.

\begin{definition}
    A \defn{partial board} is a two-dimensional rectangular array,
    where each entry is either a \defn{covered cell} or an \defn{uncovered cell}.
    A covered cell is an unknown cell which may or may not be a mine.
    An uncovered cell has an integer (and is known to not contain a mine),
    representing a Minesweeper clue.
    For a partial board \(B\),
    we denote the set of covered cells \(\C(B)\),
    and the set of uncovered cells \(\U(B)\).
\end{definition}

\begin{definition}[Minesweeper consistency problem] \label{def:cons}
    A partial board \(B\) is \defn{consistent}
    if there exists a set \(M\subseteq\C(B)\),
    representing all locations of mines,
    such that $M\cap\U(B)=\emptyset$
    and the integer in each uncovered cell \(c\in\U(B)\)
    counts the number of cells in $M$
    which are orthogonally or diagonally adjacent to \(c\).
    Otherwise, \(B\) is \defn{inconsistent}.
    See Figures~\ref{fig:consistent} and~\ref{fig:inconsistent}.
    The input to the \defn{Minesweeper consistency problem}
    is a partial board,
    and the problem asks whether it is consistent.
\end{definition}

\begin{figure}
    \centering
    \fig{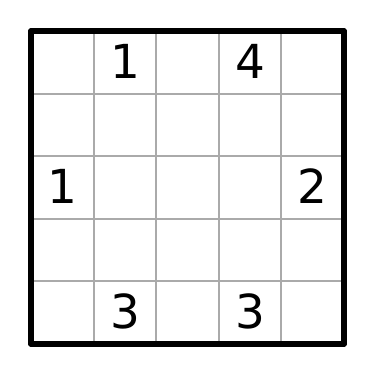}
    \fig{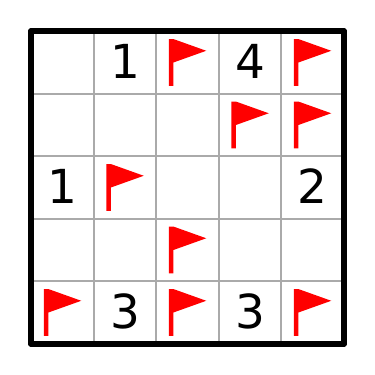}
    \caption{Left: an example of a consistent board. Right: one possible way to satisfy all clues.}
    \label{fig:consistent}
\end{figure}

\begin{figure}
    \centering
    \fig{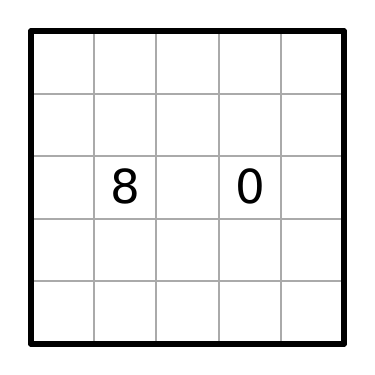}
    \caption{An example of an inconsistent board.}
    \label{fig:inconsistent}
\end{figure}

Kaye proves that the consistency problem is NP-complete \cite{cons}.
We also prove NP-completeness as a side effect of our framework.

\begin{proposition}
    The Minesweeper consistency problem is in NP.
\end{proposition}

\begin{proof}
    Given a partial board \(B\),
    a certificate of consistency is the \(M\) of Definition~\ref{def:cons}.
    We can iterate over \(\U(B)\)
    and check that each clue is satisfied
    in polynomial time.
\end{proof}

We leave the proof of hardness to Section~\ref{sec:framework}.

\subsection{Inference}

The Minesweeper inference problem was first posed by
Allan Scott, Ulrike Stege, and Iris van Rooij \cite{inf}.
Informally,
this problem asks whether a partially completed Minesweeper board
has a logical deduction available to the player
that lets them click on a cell
which is guaranteed to not have a mine.
We provide a formal definition
of a slight variation on the problem as originally posed.

\begin{definition}[Minesweeper inference problem] \label{def:inf}
    Given a partial board $B$,
    an \defn{inference} is a cell $c\in\U(B)$
    such that for all consistent arrangements of mines \(M\subseteq\C(B)\),
    we have \(c\not\in M\).
    The input to the \defn{Minesweeper inference problem}
    is a partial board,
    and it asks whether there is an inference.
    See Figures~\ref{fig:inference} and~\ref{fig:noinference}
\end{definition}

\begin{figure}
    \centering
    \fig{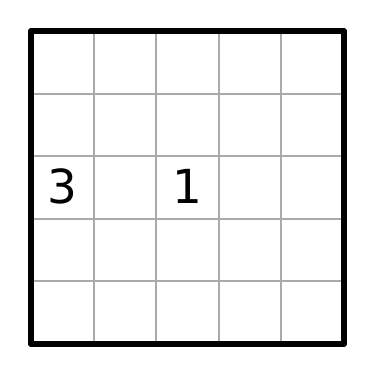}
    \fig{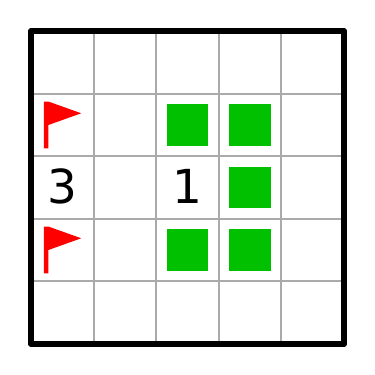}
    \caption{Left: an example of a board with an inference. Right: green squares mark cells that can be inferred.}
    \label{fig:inference}
\end{figure}

\begin{figure}
    \centering
    \fig{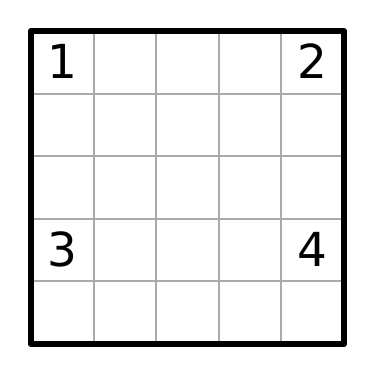}
    \caption{An example of a board that doesn't have an inference.}
    \label{fig:noinference}
\end{figure}

% This is a promise problem:
% the input is guaranteed to be consistent.
Note that this definition has two key differences
from the one originally given by Scott, Stege, and van Rooij \cite{inf}:

\begin{itemize}

    \item
        Deducing the location of a mine does not count as an inference.

    \item
        No positions of known mines are given in the input.

\end{itemize}

We prefer Definition~\ref{def:inf}
because it eliminates the need to place conditions on the input
such as ``all given mines must be deducible from the clues,''
which is otherwise necessary to avoid placing mines
that could never be deduced in a real game.
Furthermore,
the flags representing known mines can be viewed as simply a player aid---%
one could in principle play a full game of Minesweeper without marking a single mine.
Though we will proceed with Definition~\ref{def:inf} only,
we remark that all results in this paper work under either definition.

\begin{proposition}
    The Minesweeper inference problem is in coNP.
\end{proposition}

\begin{proof}
    Given a partial board \(B\),
    a certificate of noninference is,
    for each cell \(c\in\C(B)\),
    a consistent arrangement of mines \(M\)
    with \(c\in M\).
    We can iterate over the polynomially many arrangements
    given by the certificate
    and check consistency in polynomial time.
\end{proof}

Again,
the proof of hardness will be in Section~\ref{sec:framework}.

\paragraph{Error in existing proofs.}
Scott, Stege, and Rooij's proof of coNP-hardness \cite{inf}
has an issue where there is sometimes an unintended inference.
Their OR gate is shown in Figure~\ref{fig:bador}.
A cell has a mine when the literal marking it is true.
The inputs are $u$ and $v$ the output is $a$.
The 6 enforces $a+b=u+v$, and the section
at the bottom enforces $a\geq b$.
Together this forces $a=u\lor v$ and $b=u\land v$.

\begin{figure}
    \centering
    \includegraphics[scale=.3]{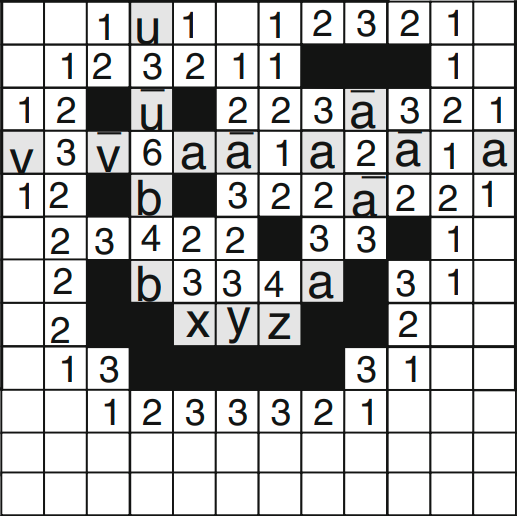}
    \caption{
      The OR gate from \cite{inf} Figure~16.
      There is a minor typo:
      the lower cell with $b$ should have $\overline b$.
    }
    \label{fig:bador}
\end{figure}

The issue occurs when we know that $u$ and $v$ can't both be true.
This might happen if the OR gate is used inside an AND gate
which merges two clauses which can't simultaneously be false,
such as if one contains $x$ and the other contains $\lnot x$.

In this case, we can deduce that $b=u\land v$ is false,
and thus the (higher) cell labeled $b$ is safe to click.
This is an inference.
One strategy for resolving this issue would be to reveal that cell
in the input, assuming one can find all inferences like this.
But this doesn't work:
That cell has either a 3 or a 4 depending on the value of $u$,
so revealing it tells the player the value of $u$,
allowing them to make more inferences and
possibly learning further information.

Thieme and Basten's more compact proof \cite{indeed}
has a very similar issue.

Our approach to avoiding this issue is twofold.
First, we design our gadgets to be `silent',
meaning clicking an inferred cell never reveals information.
This allows us to eliminate inferences by revealing those
cells in the input.
This handles unintended `local' inferences,
which we are able to detect.

Second, to prevent unintended larger-scale inferences,
we design the network of gadgets carefully.
Specifically, we ensure that for every gadget
(OR gate, crossover, etc.),
every locally valid combination of values can be achieved.
The only exception is the `final' gate, which has a forced output
when the input formula is unsatisfiable.
We maintain this property across simulations by introducing
a problem we call `promise-inference',
which also partially relaxes the constraint
that every locally valid solution is achievable.

\subsection{Solvability}

Informally,
the Minesweeper solvability problem asks whether a player can
make a sequence of inferences from a partially completed Minesweeper game
and win by clicking on all non-mine cells.
We provide a formal definition of this problem,
which to our knowledge has not been studied.

\begin{definition}[Minesweeper solvability problem] \label{def:solv}
    Let $B$ be a partial board with uncovered cells $K=\U(B)$,
    and let $M$ be a set of mines consistent with $B$.
    Note that $B$ is determined by $M$ and $K$.
    Here \(M\) represents the secret set of mines,
    which is thought of as unknown to the player,
    and \(K\) represents the set of known cells.

    Consider an ordering \(O\) of \(G\setminus(M\cup K)\),
    meaning a list containing each element of
    \(G\setminus(M\cup K)\) exactly once.
    For an element \(o\in G\setminus(M\cup K)\),
    let \(O = O_{init} \concat [o] \concat O_{tail}\),
    where \(\concat\) denotes concatenation.
    We say $M$ is \defn{solvable} from $K$
    if there exists an \(O\)
    such that for all \(o\in O\),
    \(o\) is a inference
    for the (unique) partial board consistent with $M$
    whose uncovered cells are \(K\cup O_{init}\).
    Otherwise, $M$ is \defn{unsolvable} from $K$.
    See Figure~\ref{fig:solvable}.

    The input to the \defn{Minesweeper solvability problem}
    consists of the dimensions of a rectangular grid \(G\)
    and two disjoint subsets of cells \(M\subseteq G\) and \(K\subseteq G\).
    The problem asks whether $M$ is solvable from $K$.
\end{definition}

\begin{figure}
    \centering
    \fig{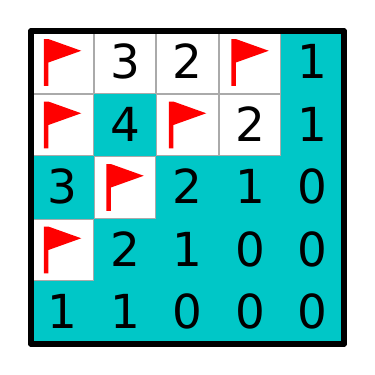}
    \caption{
        An example of a solvable board.
        Cyan cells indicate elements of \(K\),
        and white cells are unknown to the player.
        Red flags indicate elements of $M$.
        From the initial state,
        the only provably safe cell is the fourth cell in the second row.
        After clicking it,
        the player can now deduce the third cell in the first row.
        Finally,
        this reveals enough information to deduce the second cell in the first row.
    }
    \label{fig:solvable}
\end{figure}

Intuitively,
solvability simulates a player who is given a Minesweeper puzzle on their computer to solve.
If such a player wishes to guarantee that they do not click a mine
(to win without any luck,
or because they have antagonistic implementation
which repositions mines to make the player lose if they
click a cell that isn't provably safe,
such as `expert mode' in \emph{14 Minesweeper Variants}),
they must click covered cells in an order \(O\)
that satisfies Definition~\ref{def:solv}.

\begin{proposition}
    The Minesweeper solvability problem is in coNP.
\end{proposition}

\begin{proof}
    Given a game \((M,K)\),
    a certificate of unsolvability is a set \(K'\supseteq K\)
    together with a certificate of noninference for \(K'\).

    Suppose there is such a certificate.
    As long as the information available to the player is a subset of $K'$,
    they cannot infer that a cell outside $K'$ is safe.
    Any order has a first cell outside $K'$,
    which by assumption is not an inference.

    Conversely, suppose an instance is unsolvable.
    Starting from $K$, repeatedly make inferences
    and add them to the known information.
    By assumption, this must get stuck before all non-mine cells are uncovered,
    meaning we reach a point with no inferences.
    Then take $K'$ to be the uncovered cells at that point.
\end{proof}

Once again,
we prove hardness in Section~\ref{sec:framework}.

\section{The Minesweeper gadget framework} \label{sec:framework}

In this section, we describe an abstract framework which we will later apply to Minesweeper.
The wires in Minesweeper hardness proofs generally have two states, which can be thought of as two orientations of an edge, so our framework is a general form of graph orientation.

\begin{definition}
A \defn{gadget} is an abstract object which has
\begin{itemize}
  \item a finite set of \defn{ports}, in a specified cyclic order
    (we will generally list the ports in this order,
    and describe it more explicitly when relevant).
  \item a \defn{constraint}, which is a set of subsets of the ports.
\end{itemize}
\end{definition}

Gadgets will interact via directed edges connecting ports.
The constraint says which sets of edges pointing towards the gadget should be considered legal.

The gadgets we name are collected in Table~\ref{tbl:gadgets},
and will also be described as they come up.

\begin{table}
    \centering
    {\renewcommand\arraystretch{2.5}
      \begin{tabular}{lll}
        Name & Icon & Constraint \\\hline

        fixed (true) terminal
        &
        \begin{circuitikz}[baseline=(current bounding box.center)]
          \draw (0,0) node[left] {$a$} -- (.5,0) node[right] {\(F\)};
        \end{circuitikz}
        &
        \(\{\{\}\}\)

        \\

        fixed (false) terminal
        &
        \begin{circuitikz}[baseline=(current bounding box.center)]
          \draw (0,0) node[left] {$a$} -- (.5,0) node[right] {\(T\)};
        \end{circuitikz}
        &
        \(\{\{a\}\}\)

        \\

        free terminal
        &
        \begin{circuitikz}[baseline=(current bounding box.center)]
          \draw (0,0) node[left] {\(a\)} -- (.5,0) node[right] {\phantom T};
        \end{circuitikz}
        &
        \(\{\{\}, \{a\}\}\)

        \\

        AND gate
        &
        \begin{circuitikz}[baseline=(current bounding box.center)]
          \draw (0,0) node[and port] (p) {} (p.in 1) node[left] {\(a\)} (p.in 2) node[left] {\(b\)} (p.out) node[right] {\(c\)};
        \end{circuitikz}
        &
        \(\{\{c\}, \{a,c\}, \{b,c\}, \{a,b\}\}\)

        \\

        OR gate
        &
        \begin{circuitikz}[baseline=(current bounding box.center)]
          \draw (0,0) node[or port] (p) {} (p.in 1) node[left] {\(a\)} (p.in 2) node[left] {\(b\)} (p.out) node[right] {\(c\)};
        \end{circuitikz}
        &
        \(\{\{c\}, \{a\}, \{b\}, \{a,b\}\}\)

        \\

        NOT gate
        &
        \begin{circuitikz}[baseline=(current bounding box.center)]
          \draw (0,0) node[not port] (p) {} (p.in) node[left] {\(a\)} (p.out) node[right] {\(c\)};
        \end{circuitikz}
        &
        \(\{\{\}, \{a,c\}\}\)

        \\

        NOR gate
        &
        \begin{circuitikz}[baseline=(current bounding box.center)]
          \draw (0,0) node[nor port] (p) {} (p.in 1) node[left] {\(a\)} (p.in 2) node[left] {\(b\)} (p.out) node[right] {\(c\)};
        \end{circuitikz}
        &
        $\{\{\},\{a,c\},\{b,c\},\{a,b,c\}\}$
        \\

        ($k$-way) fanout gate
        &
        \begin{circuitikz}[baseline=(current bounding box.center)]
          \draw (0,0) node[
          muxdemux,
          muxdemux def={NL=4, NB=0, Lh=3, Rh=2, w=0.4, inset Lh=1},
          rotate=180
          ] (p) {} (p.rpin 1) node[left] {\(a\)} (p.lpin 1) node[right] {\(c_k\)} (p.lpin 2) node[right] {\(c_{k-1}\)} (p.lpin 3) node[right] {\(c_2\)} (p.lpin 4) node[right] {\(c_1\)} (0.7,0.15) node{\(\vdots\)};
        \end{circuitikz}
        &
        \(\{\{a\}, \{c_1,\ldots,c_k\}\}\)

        \\

        1-in-3 gadget
        &
        \begin{circuitikz}[baseline=(current bounding box.center)]
          \draw
          (-0.75,0) node[left] {\(a\)} -- (-0.25,0)
          (0,0.75) node[above] {\(b\)} -- (0,0.25)
          (0.75,0) node[right] {\(c\)} -- (0.25,0)
          (0,-0.5) node {}
          (0,0) circle (0.25) node {\sfrac13};
        \end{circuitikz}
        &
        \(\{\{a\}, \{b\}, \{c\}\}\)

        \\

        crossover
        &
        \begin{circuitikz}[baseline=(current bounding box.center)]
          \draw
          (-0.5,0) node[left] {\(a_1\)} to[crossing] (0.5,0) node[right] {\(a_2\)}
          (0,0.5) node[above] {\(b_1\)} -- (0,-0.5) node[below] {\(b_2\)};
        \end{circuitikz}
        &
        \(\{\{a_1,b_1\},\{a_2,b_1\},\{a_1,b_2\},\{a_2,b_2\}\}\)

      \end{tabular}
    }
    \caption{
      The gadgets we define in this paper,
      the icons we use to draw them (for those we show in networks),
      and their constraints.
    }
    \label{tbl:gadgets}
\end{table}

\begin{definition}
A \defn{network} of gadgets from a set $S$ is an undirected graph where
\begin{itemize}
  \item each vertex is labeled with a gadget from $S$
  \item if $G\in S$ has $k$ ports, each vertex labeled $G$ has degree $k$ and its edge incidencies are labeled in a bijection with ports of $G$.
\end{itemize}

A \defn{planar network} is such a graph equipped with a planar embedding, such that the cyclic order of edges around each vertex matches the order of the ports of the corresponding gadget.
\end{definition}

We often equivocate between vertices and gadgets, and between edge incidencies and ports---think of each vertex as a copy of its label, and think of edges as connecting ports to ports in a matching.

We draw planar networks using the icons in Table~\ref{tbl:gadgets},
which also indicate the correspondence between
edges and ports (except for when it doesn't matter by symmetry).

\begin{definition}
An \defn{assignment} to a (planar) network assigns a direction to each edge.
A vertex is \defn{satisfied} if the set of (labels of) edges pointing into it
is in its (label's) constraint.
An assignment is \defn{satisfying} if every vertex is satisfied.
\end{definition}

Planar Graph Orientation (PGO) is the study of satisfying assignments of planar networks.

\subsection{Gates as gadgets}\label{sec:pgo gates}

One important kind of gadget is logic gates,
which can be interpreted as gadgets: inputs are edges entering from the left and output are edges exiting on the right.
Interpret pointing right as true.
The gadget's constraint allows the inputs to be arbitrary,
but forces the correct outputs for each combination of inputs.
More precisely, for each subset $S$ of input ports,
the constraint contains $S\cup T$,
where $T$ is the set of outputs that are \emph{false} when precisely the inputs in $S$ are true.

It's important to keep in mind distinction between gate gadgets,
which compute a value,
and ``normal'' gadgets, which enforce a constraint.

For example, the \defn{OR gate} has ports $\{a,b,c\}$ and constraint $\{\{c\},\{a\},\{b\},\{a,b\}\}$.
This constraint allows any subset of the input ports $a$ and $b$ to have edges pointing in,
and enforces that the edge incident to the output port $c$ to point out exactly when at least one input port points in.
In other words, it computes the OR of $a$ and $b$, and outputs it at $c$.

On the other hand, the \defn{OR gadget}
(which we don't need beyond this example)
has only two ports $\{a,b\}$
and constraint $\{\{a\},\{b\},\{a,b\}\}$.
It enforces that at least one edge points in,
or that at least one input is true.
This is equivalent to an OR gate with the output forced to be ``true'' (pointing away).

Any boolean circuit can be represented as a network of gate gadgets,
by attaching input ports to output ports in the natural way.
Allow us to leave inputs and outputs to the full circuit
as ``dangling'' edges for now.
It follows from the definition of gate gadgets---formally,
one can induct on the depth of a gate---that
this network has exactly one satisfying assignment
for each combination of orientations of the dangling input edges,
and in each the orientations of the output edges
encode the output of the circuit
(and internal edges encode the internal state of the circuit).

If an output connects to $k>1$ inputs,
we need to use a $k$-way \defn{fanout gate},
which has ports $\{a,c_1,\dots,c_k\}$
and constraint $\{\{a\},\{c_1,\dots,c_k\}\}$.
This is the gadget representing the gate that duplicates its input $k$ times.

A drawing of a circuit in the plane becomes a planar network of gate gadgets.
A crossing pair of wires is translated to the \defn{crossover gate},
which has two inputs and two outputs which match the inputs
in the opposite order.
As a gadget, the crossover gate has ports $\{a_1,b_1,a_2,b_2\}$,
importantly in that cyclic order.
It has constraint $\{\{a_1,b_1\},\{a_1,b_2\},\{a_2,b_1\},\{a_2,b_2\}\}$.
One can think of it as two wires $a_1\to a_2$ and $b_1\to b_2$
that cross in a circuit.
However, the directions of these wires are arbitrary
because the gadget is highly symmetric.
As a gadget outside the perspective of networks of logic gates,
the crossover gate is two edges that can independently be assigned
orientations, and which cross each other.
So we will also call it simply the \defn{crossover}.

\subsection{Decision problems}

We consider three decision problems about planar graph orientation,
corresponding to the Minesweeper decision problems we're interested in.
Each of them is parameterized by a set of gadgets $S$---throughout
we consider only finite sets of gadgets---and takes as input a planar network $N$ of gadgets from $S$.

The problem corresponding to Minesweeper consistency is straightforward.

\begin{problem}
PGO \defn{consistency} with $S$ asks whether $N$ has a satisfying assignment.
\end{problem}

For Minesweeper inference, we need to avoid a subtle issue,
which is the error in prior claims of coNP-hardness \cite{inf,indeed}.
If there is a gadget for which we can deduce that
some legal combination of edge orientations
can't be extended to a full satisfying assignment,
this information may allow us to infer the value
of a Minesweeper cell internal to the gadget.
This can happen even if we can't deduce the orientation of any particular
edge,
so the most obvious PGO inference problem fails to reduce to Minesweeper,
and thus isn't useful.

Our strategy for resolving this issue will ultimately be to ``click'' all
cells in the Minesweeper instance that could be deduced in this way.
To make this work, we will need the values of those cells
to not reveal any additional information
(a property we will call ``silence''),
and we need the reduction to find all such cells in polynomial time.
For our gadget framework to help here,
we need to define the inference problem carefully,
and in particular it needs to be aware of which combinations of edges
are ruled out by ``semilocal'' deductions.

\begin{problem}
PGO \defn{promise-inference} with $S$ is given a network $N$
as well as, for each vertex in $N$, a \defn{semilocal constraint}
which is a subset of its constraint.
We say an edge $e$ is \defn{locally forced} if the semilocal constraint
of one of its vertices requires it has a particular orientation
(because it's either in every element or in no element),
and \defn{locally free} otherwise.

We are promised that
\begin{itemize}
  \item in every satisfying assignment, the set of edges pointing into a vertex is in its semilocal constraint; and
  \item either
  \begin{itemize}
    \item there is a locally free edge to which every satisfying assignment assigns the same direction; or
    \item for each vertex $v$ and set of edges in its semilocal constraint, there's a satisfying assignment that makes exactly those edges point towards $v$.
  \end{itemize}
\end{itemize}
We are asked to determine whether there is an edge whose orientation can be inferred but isn't immediate from semilocal constraints;
that is, whether we are in the first option above.
\end{problem}

Note that the two options can't simultaneously be true,
since the existence of such an edge would imply that its incident vertices' semilocal constraints have elements that can't be extended to satisfying assignments.

The intent of semilocal constraints is to be between the levels of individual gadgets' ``local'' constraints
and ``global'' deductions that require understanding the entire network.
A semilocal constraint typically contains the sets of in-pointing edges
that are attainable when looking at some constant-size neighborhood of a gadget.

The first part of the promise ensures that semilocal constraints are actually enforced by the structure of the network.
In the case where there's no inferable edge,
the second part says that it isn't possible to deduce more about
the immediate neighborhood of a gadget than its semilocal constraint.

Finally, the PGO decision problem analogous to Minesweeper solvability is simpler.

\begin{problem}
PGO \defn{uniqueness} with $S$ is given $N$ as well as a satisfying assignment of $N$,
and asks whether it is the unique satisfying assignment.
\end{problem}

PGO uniqueness is related to Minesweeper solvability because it's possible to deduce the orientations of all edges if and only if there's a unique satisfying assignment.

\begin{lemma}\label{lem:pgo containment}
For any set $S$ of gadgets,
\begin{enumerate}
  \item PGO consistency with $S$ is in NP.
  \item PGO promise-inference with $S$ is in promise-coNP.
  \item PGO uniqueness with $S$ is in coNP.
\end{enumerate}
\end{lemma}

\begin{proof}
Each problem has a straightforward certificate:
\begin{enumerate}
  \item A satisfying assignment serves as a certificate of consistency.
  \item A certificate that there is no inference consists of, for locally free edge and each orientation, a satisfying assignment that assigns the orientation to the edge.
  \item A second satisfying assignment serves as a certificate of nonuniqueness.
\end{enumerate}
\end{proof}

\subsection{Hardness}\label{sec:pgo hardness}

We now prove hardness of each PGO decision problem for the appropriate class,
with a specific set of gadgets.
The gadget sets are chosen to make the hardness proofs simple;
there are easy ways to reduce the number of different gadgets needed,
and we will further simplify our gadget sets
using simulation in Section~\ref{sec:pgo simplify}.
With Lemma~\ref{lem:pgo containment}, we have completeness in each case.

\begin{theorem}
PGO satisfiability with
free terminals,
fanout gates,
and 1-in-3 gadgets
is NP-hard.
\end{theorem}

Some of these gadgets are new:
the \defn{1-in-3 gadget} has three ports,
and its constraint says that exactly one of them must point in.
The \defn{free terminal} has one port,
which is allowed to point in either direction.

\begin{proof}
  We reduce from planar positive 1-in-3SAT,
  which is NP-hard \cite{1in3}.
  Each variable with $k$ occurrences becomes
  a free terminal that leads to a $k$-way fanout gate.
  Each clause becomes a 1-in-3 gadget.
  We connect the outputs of the fanout gates
  to 1-in-3 gadgets as in the 1-in-3SAT formula,
  which is planar.
  Satisfying assignments of this network
  correspond to satisfying assignments of the formula.
\end{proof}

\begin{theorem}
PGO promise-inference with
free terminals,
fanout gates,
crossovers,
OR gates,
and AND gates
is coNP-hard.
\end{theorem}

\begin{proof}

    \begin{figure}
        \centering
        \newcommand\sep{2.5}
        \newcommand\xsh{0.7}
        \newcommand\ysh{1.8}
        \newcommand\vsh{0.5}
        \newcommand\ns{5.5}

        \begin{circuitikz}
            \foreach \side/\mul/\porta/\portb/\pos in {n/-1/1/2/above, s/1/2/1/below} {

                \draw (\sep*3+\xsh*0.5,{-(\ns-\ysh*0.5)*\mul}) node {\(\cdots\)};
                \draw (\sep*3+\xsh*0.5,{-(\vsh)*\mul}) node {\(\cdots\)};

                \foreach \n in {1,2,4} {
                    \draw[shift={(\sep*\n,-\ns*\mul)},yscale=\mul]
                    % clause
                    (\xsh,0) node[or port, rotate=-90*\mul] (\side or2\n) {}
                    (0,\ysh) node[or port, rotate=-90*\mul] (\side or1\n) {}
                    (\side or1\n.out) |- (\side or2\n.in \porta)
                    % variable
                    (\xsh*0.5,\ns-\vsh) node[
                    muxdemux,
                    muxdemux def={NL=4, NB=0, Lh=3, Rh=2, w=0.4, inset Lh=1},
                    rotate=90*\mul
                    ] (\side var\n) {}
                    node[yshift=-8.5*\mul] {\(\cdots\)}
                    ;
                }

                % and together all the ors, add a free variable
                \draw
                    (\sep*3,{-(\ns+\ysh)*\mul}) node[and port] (\side and1) {}
                    (\sep*5,{-(\ns+\ysh)*\mul}) node[and port] (\side and2) {}
                    (\sep*5.5,{-(\ns-\ysh)*\mul}) node[and port, rotate=90*\mul] (\side and3) {}
                    ;

                % wire everything up
                \draw
                    (\side or24.out) |- (\side and2.in \portb)
                    (\side and1.out) -- ++(0.5,0) ++(0.5,0) node{\(\cdots\)} ++(0.5,0) -- ++(0.5,0) |- (\side and2.in \porta)
                    (\side or22.out) |- (\side and1.in \portb)
                    (\side or21.out) |- (\side and1.in \porta)
                    (\side and2.out) -| (\side and3.in \portb)
                    (\side and3.in \porta) node[\pos] {\(z_\porta\)}
                    ;

            }

            % connect variable sides
            \foreach \n in {1,2,4} \draw (nvar\n.rpin 1) -- (svar\n.rpin 1);

            % final output
            \draw (\sep*6,0) node[and port] (and) {};
            \draw (nand3.out) |- (and.in 1) (sand3.out) |- (and.in 2) (and.out) node[right] {\(r\)};

            % example wires
            \draw
                (svar1.lpin 4 -| sor12.in 1) ++(0,-1.2) node[jump crossing] (jc) {}
                (svar1.lpin 4) |- (jc.west) (jc.east) -| (sor14.in 2)
                (svar4.lpin 1) -- ++(0,-0.75) -| (jc.north) (jc.south) -- (sor12.in 1)
                (nvar2.lpin 1) -- ++(0,1) -| (nor14.in 1)
                (nvar2.lpin 3) -- ++(0,1) -| (nor12.in 2)
                ;

            % example undrawn lines
            \draw[densely dashed]
                (nor11.in 1) -- ++(0,-0.5) node(bottom){}
                (nor11.in 2) -- ++(0,-0.5)
                (nor21.in 2) -- (bottom -| nor21.in 2)
                (nor12.in 1) -- ++(0,-0.5)
                (nor22.in 2) -- (bottom -| nor22.in 2)
                (nor14.in 2) -- ++(0,-0.5)
                (nor24.in 2) -- (bottom -| nor24.in 2)
                (sor11.in 2) -- ++(0,0.5) node(top){}
                (sor11.in 1) -- ++(0,0.5)
                (sor21.in 1) -- (top -| sor21.in 1)
                (sor12.in 2) -- ++(0,0.5)
                (sor22.in 1) -- (top -| sor22.in 1)
                (sor14.in 1) -- ++(0,0.5)
                (sor24.in 1) -- (top -| sor24.in 1)
                (nvar4.lpin 1) -- ++(0,0.5)
                (nvar4.lpin 2) -- ++(0,0.5)
                (nvar4.lpin 3) -- ++(0,0.5)
                (nvar4.lpin 4) -- ++(0,0.5)
                (svar4.lpin 2) -- ++(0,-0.5)
                (svar4.lpin 3) -- ++(0,-0.5)
                (svar4.lpin 4) -- ++(0,-0.5)
                (nvar2.lpin 2) -- ++(0,0.5)
                (nvar2.lpin 4) -- ++(0,0.5)
                (svar2.lpin 1) -- ++(0,-0.5)
                (svar2.lpin 2) -- ++(0,-0.5)
                (svar2.lpin 3) -- ++(0,-0.5)
                (svar2.lpin 4) -- ++(0,-0.5)
                (nvar1.lpin 1) -- ++(0,0.5)
                (nvar1.lpin 2) -- ++(0,0.5)
                (nvar1.lpin 3) -- ++(0,0.5)
                (nvar1.lpin 4) -- ++(0,0.5)
                (svar1.lpin 1) -- ++(0,-0.5)
                (svar1.lpin 2) -- ++(0,-0.5)
                (svar1.lpin 3) -- ++(0,-0.5)
                ;

            % example boxes
            \draw[draw=red] ($(nor11.north west)+(-0.3,0)$) rectangle ($(nor21.south east)+(0.3,0)$);
            \draw[draw=green] ($(nvar4.south west)+(-0.3,0.5)$) rectangle ($(svar4.south west)+(0.3,-0.5)$);
        \end{circuitikz}
        \caption{Our reduction for coNP-hardness of PGO promise-inference. See Table 1 for gadget notation.}
        \label{fig:pgo promise inference}
    \end{figure}
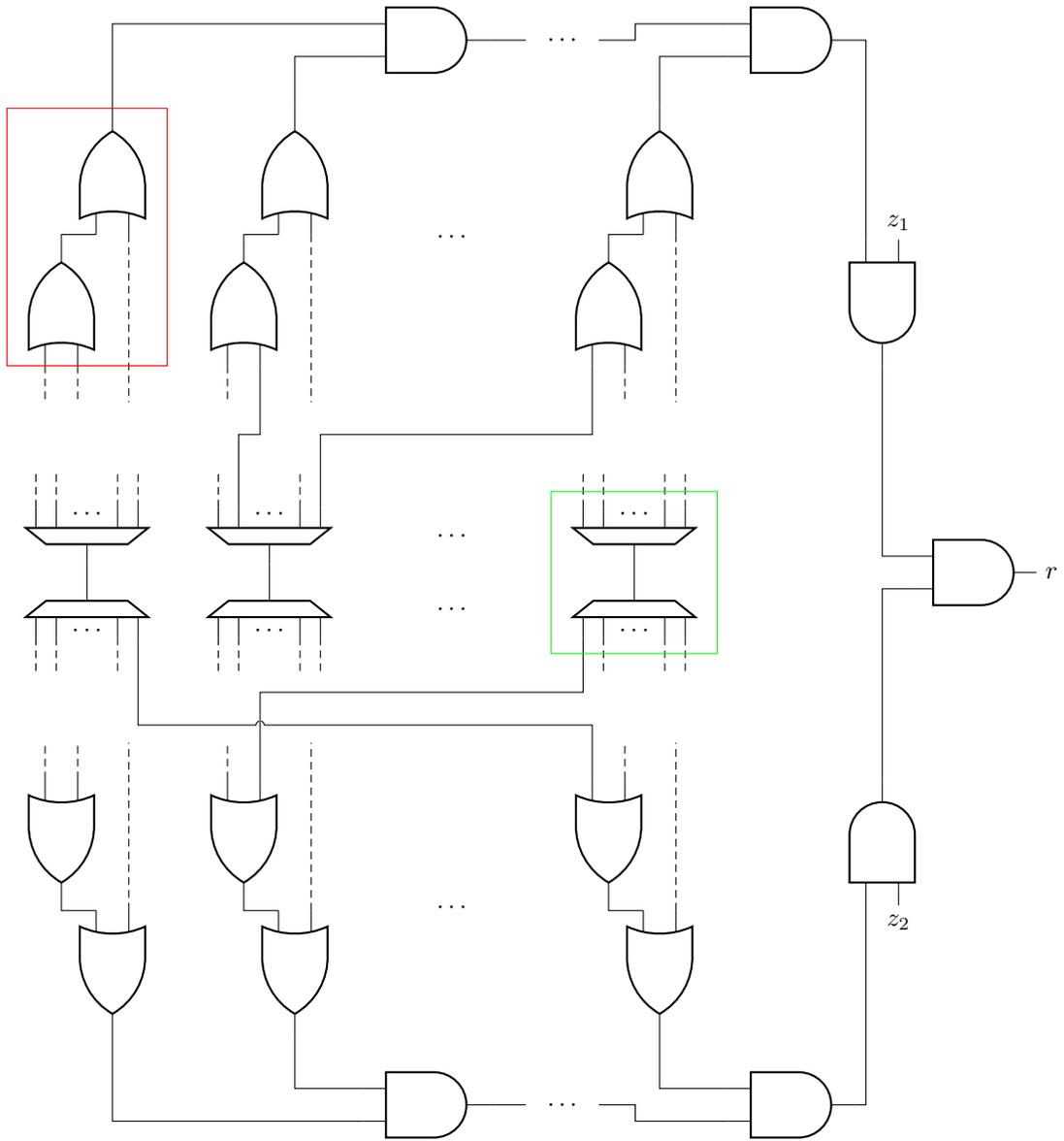

    We reduce from the complement of monotone 3SAT,
    which is NP-hard \cite{monotone3sat}.
    The layout of gadgets is depicted in Figure~\ref{fig:pgo promise inference}.
    The green outline is an example of a variable,
    and the red outline is an example of a clause.
    We think of the top half of the circuit as the positive clauses,
    and the bottom half of the circuit as the negative clauses.
    We remove duplicate clauses before constructing the circuit.
    Each gadget has as its semilocal constraint its full constraint set.

    The output of the formula is given by \(r\),
    which is connected to a free terminal.
    If the formula is unsatisfiable,
    there is an inference;
    namely,
    \(r\) is false
    (points towards the AND gate).
    So we just need to show that if the formula is satisfiable,
    there is no inference; that is,
    for each gadget in the reduction,
    every set of edges in its constraint
    is possible to achieve (pointing in) in some consistent assignment.
    Because each gadget is a gate,
    this is equivalent to every combination of input values being attainable.

    There is a unique consistent assignment
    for each choice of values for variables, $z_1$, and $z_2$.
    We will describe such a choice that achieves each combination of input values
    for each gadget.

    Each fanout depends on a single variable, which can have either value.
    Each crossover is between two wires connected to different variables,
    so all combinations are possible.
    The inputs of each OR gate in a clause can be chosen by setting the input variables as desired.

    Consider any of the AND gates combining the outputs of the clauses on one half of the circuit,
    and assume for simplicity that it's on the positive half.
    We can make both inputs true (pointing in) or both inputs false (pointing out)
    by setting all variables true
    or all variables false, respectively.
    To make exactly one input true,
    note that we can choose a single (positive) clause
    to be unsatisfied:
    make the three variables in that clause false and all other variables true.
    Then the clause is unsatisfied, but (since clauses have exactly three literals and
    there aren't duplicates), every other positive clause is satisfied.
    By choosing a clause that feeds into either side of the AND gate in question,
    we can make either of its inputs false and the other true.
    The same argument shows that
    the AND gates with $z_1$ and $z_2$ can also have any combination of inputs.

    Finally,
    consider the rightmost AND gate, with output $r$.
    By hypothesis,
    the formula is satisfiable,
    so both inputs can be made true
    by satisfying the formula and setting \(z_1\) and \(z_2\) to both be true.
    We can then independently choose to make either input false
    by flipping \(z_1\) or \(z_2\) as appropriate.
\end{proof}

\begin{theorem}
PGO uniqueness with
free terminals,
fanout gates,
crossovers,
NOT gates,
OR gates,
AND gates,
and fixed terminals
is coNP-hard.
\end{theorem}

The \defn{fixed terminal} has one port,
and its constraint forces the port's value.
There are two kinds of fixed terminal;
we will use the \defn{fixed true terminal},
which forces the edge to point in.

\begin{proof}

  \begin{figure}
    \centering
    \newcommand\sep{2.5}
    \newcommand\xsh{0.7}
    \newcommand\ysh{1.8}
    \newcommand\vsh{0.5}
    \newcommand\ns{5.5}

    \begin{circuitikz}
      \foreach \side/\mul/\porta/\portb/\pos in {s/1/2/1/below} {

        \draw (\sep*3+\xsh*0.5,{-(\ns-\ysh*0.5)*\mul}) node {\(\cdots\)};
        \draw (\sep*3+\xsh*0.5,{-(\vsh)*\mul}) node {\(\cdots\)};

        \foreach \n/\v in {1/1,2/2,4/n} {
          \draw[shift={(\sep*\n,-\ns*\mul)},yscale=\mul]
          % clause
          (\xsh,0) node[or port, rotate=-90*\mul] (\side or2\n) {}
          (0,\ysh) node[or port, rotate=-90*\mul] (\side or1\n) {}
          (\side or1\n.out) |- (\side or2\n.in \porta)
          % variable
          (\xsh*0.5,\ns-\vsh) node[
          muxdemux,
          muxdemux def={NL=4, NB=0, Lh=3, Rh=2, w=0.4, inset Lh=1},
          rotate=90*\mul
          ] (\side var\n) {}
          node[yshift=-8.5*\mul] {\(\cdots\)}
          (\side var\n.rpin 1) node[above left] {\(x_\v\)}
          ;
        }

        % and together all the ors, add a free variable
        \draw
        (\sep*3,{-(\ns+\ysh)*\mul}) node[and port] (\side and1) {}
        (\sep*5,{-(\ns+\ysh)*\mul}) node[and port] (\side and2) {}
        (\sep*5.5,{-(\ns-\ysh)*\mul}) node[and port, rotate=90*\mul] (\side and3) {}
        (\sep*3,{-(\ns+\ysh)*\mul+10.5}) node[and port] (\side and1v) {}
        (\sep*5,{-(\ns+\ysh)*\mul+10.5}) node[and port] (\side and2v) {}
        (\sep*5.5,{-(\ns-\ysh)*\mul+5}) node[and port, rotate=-90*\mul] (\side and3v) {}
        ;

        % wire everything up
        \draw
        (\side or24.out) |- (\side and2.in \portb)
        (\side and1.out) -- ++(0.5,0) ++(0.5,0) node{\(\cdots\)} ++(0.5,0) -- ++(0.5,0) |- (\side and2.in \porta)
        (\side and1v.out) -- ++(0.5,0) ++(0.5,0) node{\(\cdots\)} ++(0.5,0) -- ++(0.5,0) |- (\side and2v.in \portb)
        (\side or22.out) |- (\side and1.in \portb)
        (\side or21.out) |- (\side and1.in \porta)
        (\side and2.out) -| (\side and3.in \porta)
        (\side and2v.out) -| (\side and3v.in \portb)
        ;

      }

      \foreach \v/\g/\p in {1/1/1,2/1/2,4/2/2} {
        \draw (svar\v.rpin 1) |- (sand\g v.in \p);
      }
      \draw (sand3v.in 2) -- ++(-0.8,0) node(tmp){} -- (sand3.in 1 -| tmp) node[midway,left] {\(z\)} -- (sand3.in 1);

      % final output
      \draw (\sep*6,0) node[or port] (and) {};
      \draw (sand3v.out) |- (and.in 1) (sand3.out) |- (and.in 2) (and.out) node[right] {\(T\)};

      % example wires
      \draw
      (svar1.lpin 4 -| sor12.in 1) ++(0,-1.25) node[jump crossing] (jc) {}
      (svar1.lpin 4) |- (jc.west) (jc.east) -| (sor14.in 2)
      (svar4.lpin 2) ++(0,-0.65) node(notin){}
      ($(jc.north |- notin)!0.5!(notin)$) node(not)[not port,rotate=180] {}
      (svar4.lpin 2) |- (not.in)
      (not.out) -| (jc.north) (jc.south) -- (sor12.in 1)
      ;

      % example undrawn lines
      \draw[densely dashed]
      (sor11.in 2) -- ++(0,0.5) node(top){}
      (sor11.in 1) -- ++(0,0.5)
      (sor21.in 1) -- (top -| sor21.in 1)
      (sor12.in 2) -- ++(0,0.5)
      (sor22.in 1) -- (top -| sor22.in 1)
      (sor14.in 1) -- ++(0,0.5)
      (sor24.in 1) -- (top -| sor24.in 1)
      (svar4.lpin 1) -- ++(0,-0.5)
      (svar4.lpin 3) -- ++(0,-0.5)
      (svar4.lpin 4) -- ++(0,-0.5)
      (svar2.lpin 1) -- ++(0,-0.5)
      (svar2.lpin 2) -- ++(0,-0.5)
      (svar2.lpin 3) -- ++(0,-0.5)
      (svar2.lpin 4) -- ++(0,-0.5)
      (svar1.lpin 1) -- ++(0,-0.5)
      (svar1.lpin 2) -- ++(0,-0.5)
      (svar1.lpin 3) -- ++(0,-0.5)
      ;
    \end{circuitikz}
    \caption{
      Our reduction for coNP-hardness of PGO uniqueness.
      See Table~\ref{tbl:gadgets} for gadget notation.
    }
    \label{fig:pgo uniqueness}
  \end{figure}
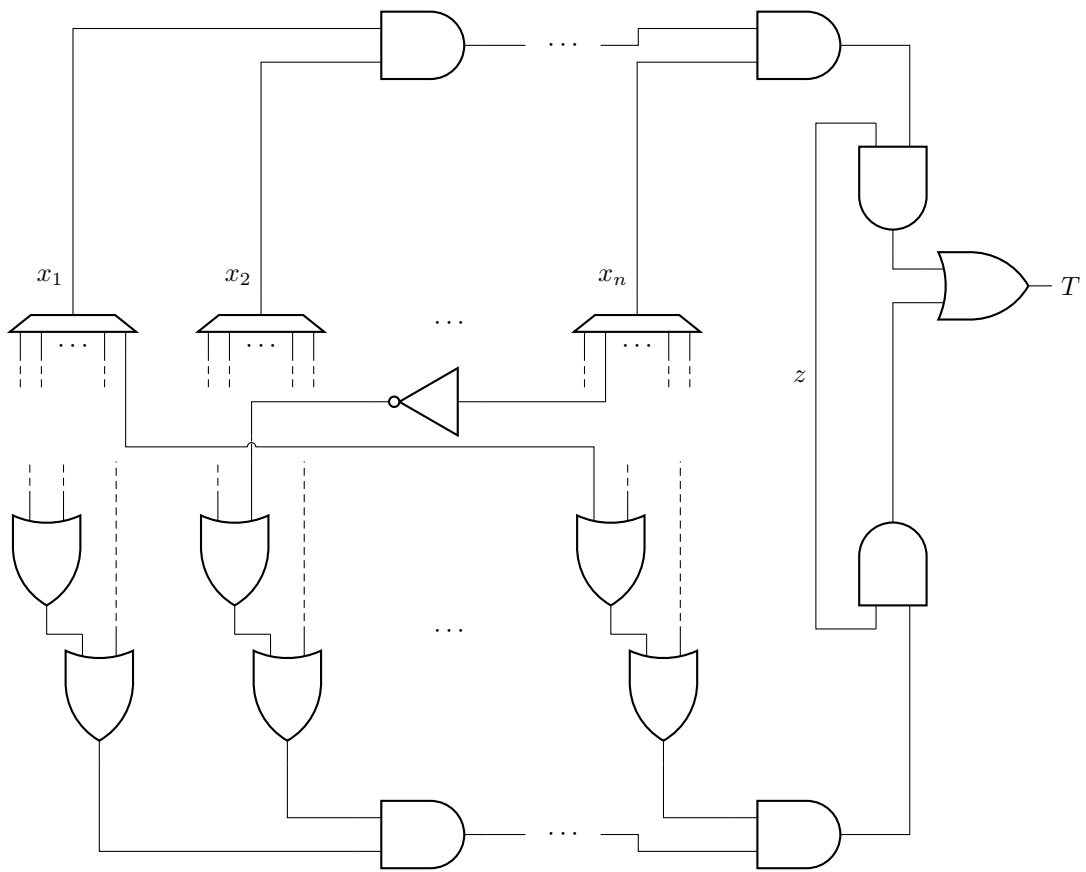

  We reduce from the complement of 3SAT,
  which is NP-hard \cite{sat}.
  Refer to Figure~\ref{fig:pgo uniqueness}.
  Each variable becomes an edge connected to a fanout gate pointing down:
  that edge pointing down represents the variable being true.
  Each clause becomes a pair of OR gates
  connected in the natural way.
  We place edges connecting variables to clauses
  in the structure of the formula.
  If edges cross, we use a crossover gadget.
  For negated literals, we place a NOT gate along the edge.

  The outputs of the clauses are merged with AND gates,
  so the edge in the bottom right of Figure~\ref{fig:pgo uniqueness}
  points right (and up) exactly when the assignment
  (based on the orientations of edges representing variables)
  is satisfying.

  In the top section of Figure~\ref{fig:pgo uniqueness},
  we merge the other ends of the variable edges with AND gates.
  The top right edge points right (and down)
  exactly when all variables are false.

  On the right, there is an edge $z$ which points towards one of two AND gates,
  with the results merged by an OR gate and then run into a fixed true terminal.
  For the output of that OR gate to be true (point right),
  either
  \begin{itemize}
    \item $z$ points up, and all variables are false; or
    \item $z$ points down, and the 3SAT formula is satisfied.
  \end{itemize}

  Note that the orientations of all edges
  are uniquely determined by those of $z$ and variables,
  even ignoring the fixed terminal.
  In particular, the network has exactly one satisfying assignment
  (with $z$ pointing down) for each satisfying assignment of the 3SAT formula,
  plus exactly one more,
  which has $z$ pointing up and all variable edges pointing up.

  The input to PGO uniqueness is the network described above
  and the satisfying assignment with $z$ pointing up.
  This is the unique satisfying assignment exactly
  when the 3SAT formula is not satisfiable.
\end{proof}

\subsection{Simulation}\label{sec:pgo simulation}

A key feature of our framework is that it allows us to abstractly construct gadgets out of other gadgets, greatly reducing the complexity of the gadgets we need to actually implement in Minesweeper.
In particular, the results in Section~\ref{sec:pgo hardness}
use many gadgets, some of which are hard to build directly,
particularly with the properties our hardness proofs require.

\begin{definition}
A \defn{simulation} using gadgets from $S$
is a network of gadgets from $S$,
except it may have some `dangling' edges incident to only one vertex.
Equivalently, the graph contains one special vertex called the `outside world' (which has the trivial constraint).

For \defn{planar} simulations,
we require that the dangling edges are in the external face,
or equivalently that the graph including the outside world is planar.
\end{definition}

See Section~\ref{sec:pgo simplify}
for several examples of simulations.

\begin{definition}
Given a simulation, the \defn{simulated gadget}
has dangling edges as ports,
and its constraint contains each set of dangling edges
for which there is a satisfying assignment
making exactly those edges point into the simulation
(away from the outside world).

In the planar case, the order of the ports is the order of the dangling edges around the simulation, or the reverse of their order around the outside world.
\end{definition}

\begin{definition}
We say $S$ \defn{simulates} $G$ if
there is a simulation using gadgets from $S$
where the simulated gadget is $G$.
For a set $T$, we say that $S$ \defn{simulates} $T$
if $S$ simulates each gadget in $T$.
\end{definition}

For PGO uniqueness and Minesweeper solvability,
we will need our simulations to be even better behaved.

\begin{definition}
A simulation of $G$ is \defn{parsimonious}
if for each legal configuration of the edges of $G$,
there is exactly one satisfying assignment of the simulation
that orients the dangling edges that way.

We say that $S$ \defn{parsimoniously} simulates $G$
if the relevant simulation is parsimonious,
and $S$ \defn{parsimoniously} simulates $T$
if $S$ parsimoniously simulates each element of $T$.
\end{definition}

Much of the point of having a theory of simulations
is that they can be composed.
This reduces conceptual complexity by letting us
break down complicated simulations into a sequence of simpler ones.

\begin{lemma}\label{lem:compose sim}
If $S$ simulates $T$ and $T$ simulates $G$, then $S$ simulates $G$.
Moreover, this composition preserves parsimony.
\end{lemma}

\begin{proof}
In the simulation of $G$ using gadgets from $T$,
replace each gadget with its simulation using gadgets from $S$.
In any satisfying assignment of the new simulation,
each component simulation is consistent with the gadget it's simulating,
so we can construct a satisfying assignment of the original simulation
with the same orientations for dangling edges
by looking at only the edges between simulated gadgets.

Conversely, any satisfying assignment of the original simulation
can be extended to one of the new simulation
by filling in each simulated gadget with an appropriate satisfying assignment.
Thus we have a simulation of $G$ using gadgets from $S$.

If each gadget is replaced with a parsimonious simulation,
there is only one way to fill in simulated gadgets this way.
So we have defined a bijection between satisfying assignments
of the original and new simulations of $G$.
If the original is parsimonious, so is the new simulation.
\end{proof}

The other half of the point of simulations is to simplify hardness proofs,
so we need them to preserve hardness.
This is straightforward for satisfiability.

\begin{lemma}
If $S$ simulates $T$,
then there is a polynomial-time reduction
from PGO satisfiability with $T$
to PGO satisfiability with $S$.
\end{lemma}

\begin{proof}
Given a network $N$ of gadgets from $T$,
replace each gadget with a (constant-size) simulation
using gadgets from $S$ to construct a network $N'$
of gadgets from $S$.

If there is a satisfying assignment of $N'$,
looking only at the edges connecting simulated gadgets
gives a satisfying assignment of $N$.

If there is a satisfying assignment of $N$,
we can set the edges connecting simulated gadgets to match it,
and then each simulated gadget has a local solution
compatible with those edges.
\end{proof}

This is somewhat more complicated for promise-inference,
but it's not so bad with the right definition for the decision problem.

\begin{lemma}
If $S$ simulates $T$,
then there is a polynomial-time reduction from
PGO promise-inference with $T$
to PGO promise-inference with $S$
\end{lemma}

\begin{proof}
We are given an instance of promise-inference with $T$,
which is a network $N$ of gadgets from $T$
and semilocal constraints.
Construct a network $N'$ of gadgets from $S$
in the same way as above.

To complete the instance of PGO promise-inference with $S$,
we must define semilocal constraints for the vertices in $N'$.
Consider a vertex $v$, which is inside a simulation of $G\in T$.
The semilocal constraint of $v$ shall contain
the sets the of edges that point into $v$
in satisfying assignments of the simulation that are compatible with
the semilocal constraint of $G$ (as a vertex of $N$).
These semilocal constraints can be computed in polynomial time
because each one requires considering only a simulation,
and simulations have constant size.

It remains to check that this instance satisfies the promise,
and falls into the same option as the input instance.
For the first part of the promise,
consider a vertex $v$ in $N'$ and a satisfying assignment.
Then $v$ is inside a simulation of some $G\in T$,
the corresponding assignment of $N$
(which has only edges between simulated gadgets)
is compatible with the semilocal constraint of $G$.
So we have a satisfying assignment
of the simulation of $G$ compatible with
the semilocal constraint of $G$,
and thus by construction it is
compatible with the semilocal constraint of $v$.

Suppose $N$ has a locally free edge whose orientation is forced.
Since every satisfying assignment of $N'$
contains a satisfying assignment of $N$
in the inter-simulation edges,
the corresponding edge of $N'$ also has forced orientation.
It must also be locally unforced:
each orientation is compatible with the semilocal constraints
of its vertices in $N$,
and therefore each orientation is compatible with
some appropriate satisfying assignment of
the simulations of its vertices.
Hence $N'$ also has a locally free edge with forced orientation.

Now suppose $N$ falls into the second option,
namely there are satisfying assignments achieving
every element of its semilocal constraints.
Consider a vertex $v$ in $N'$ which is inside a simulation of $G$,
and consider a set $L$ in the semilocal constraint of $v$.
By definition, there is a satisfying assignment
of the simulation which makes exactly the edges in $L$ point towards $v$,
and which is compatible with the semilocal constraint of $G$
as a vertex of $N$.
By assumption, there is a satisfying assignment of $N$
which orients the dangling edges of the simulation in the same way.
Combining these, and filling in the assignment for other simulations,
we obtain a satisfying assignment of $N'$
which directs exactly the edges in $L$ towards $v$, as desired.
\end{proof}

To prove an analogous result for PGO uniqueness,
we need our simulations to preserve unique solutions.
That is, they should be parsimonious.

\begin{lemma}\label{lem:pgo sim uniqueness}
If $S$ parsimoniously simulates $T$,
then there is a polynomial-time reduction from
PGO uniqueness with $T$
to PGO uniqueness with $S$
\end{lemma}

\begin{proof}
We are given a network $N$ of gadgets from $T$
and a satisfying assignment $A$.
Construct a network $N'$ of gadgets from $S$
in the same way as the two previous proofs---replace
each gadget with its simulation.

Our instance of PGO uniqueness with $S$
also needs a satisfying assignment of $N'$.
To construct one, direct inter-simulation edges to match $A$,
and extend this to a full satisfying assignment $A'$
by consistently orienting edges inside simulations.
Since the simulations are parsimonious,
there is a unique way way to do this for each simulation,
and since simulations are constant-size,
$A'$ can be computed in polynomial time.

As before, there is a correspondence between
satisfying assignments of $N$
and satisfying assignments of $N'$.
This time, however,
parsimony ensures that the correspondence is a bijection,
as illustrated by the well-definedness of $A'$.
In particular, there is a satisfying assignment of $N$ other than $A$
if and only if there is a satisfying assignment of $N'$ other than $A'$.
\end{proof}

\subsection{Simpler gadget sets}\label{sec:pgo simplify}

Now we put the theory of simulations into practice:
we will demonstrate several simulations
and use them with the results of Section~\ref{sec:pgo simulation}
to improve the results of Section~\ref{sec:pgo hardness}
to use more convenient sets of gadgets.
The simulations we use are summarized in
Figure~\ref{fig:sims}.

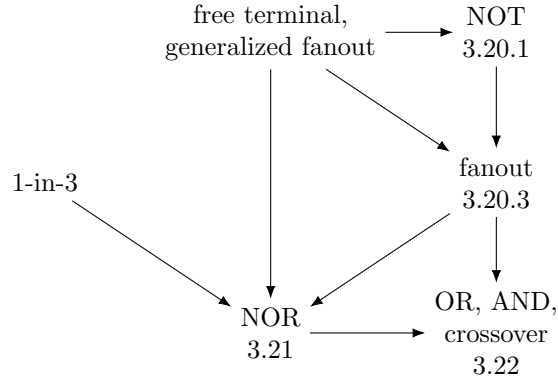
\begin{figure}
  \centering
  \begin{tikzpicture}

    \node[align=center] (fre) at (0,2) {free terminal,\\generalized fanout};
    \node[align=center] (not) at (3,2) {NOT\\\ref{lem:simfan}.1};
    \node[align=center] (fan) at (3,0) {fanout\\\ref{lem:simfan}.3};
    \node[align=center] (nor) at (0,-2) {NOR\\\ref{lem:simnor}};
    \node[align=center] (stf) at (3,-2) {OR, AND,\\crossover\\\ref{lem:simall}};
    \node[align=center] (cls) at (-3,0) {1-in-3};

    \draw[-{Latex}] (fre) -- (not);
    \draw[-{Latex}] (fre) -- (fan);
    \draw[-{Latex}] (not) -- (fan);
    \draw[-{Latex}] (fre) -- (nor);
    \draw[-{Latex}] (fan) -- (stf);
    \draw[-{Latex}] (nor) -- (stf);
    \draw[-{Latex}] (fan) -- (nor);
    \draw[-{Latex}] (cls) -- (nor);

  \end{tikzpicture}
  \caption{
    The (parsimonious) simulations
    we use to simplify our gadget set.
    Each gadget is simulated by the collection
    of gadgets pointing towards it,
    except that we will need to build
    1-in-3 gadgets, free terminals, and any generalized fanout directly.
  }
  \label{fig:sims}
\end{figure}

When we build gadgets in variants of Minesweeper,
we will construct gadgets that are like the fanout gate,
but may have more than three ports
and may have some ports reversed.

\begin{definition}
  A \defn{generalized fanout}
  is a gadget with at least three ports
  whose constraint has exactly two sets,
  which are compliments.
  That is, it has two legal configurations,
  which differ by flipping all edges.
\end{definition}

For instance, the fanout gate is a generalized fanout,
and the NOT gate is almost a generalized fanout,
except it has only two ports.

\begin{lemma}\label{lem:simfan}
  Let $F$ be a generalized fanout.
  Then $F$ and free terminals parsimoniously simulate
  \begin{enumerate}
    \item the NOT gate
    \item any generalized fanout $F'$
    \item fanout gates (with any number of outputs).
  \end{enumerate}
\end{lemma}

\begin{proof}
  We describe each simulation.
  \begin{enumerate}
    \item Let $F$ have ports $P$ and constraint $\{T,F\}$, where $T\sqcup F=P$.
      Since $|P|\geq 3$, at least one of $T$ and $F$ has size at least $2$;
      assume without loss of generality that $|T|\geq 2$, and $a,b\in T$.
      Attach free terminals to all ports of $F$ other than $a$ and $b$.
      This simulation has two satisfying assignments,
      corresponding to $T$ and $F$.
      The simulated gadget has ports $a$ and $b$,
      and constraint $\{\{a,b\},\emptyset\}$, so it is the NOT gate.
    \item Connect copies of $F$ in a tree
      until there are at least as many dangling edges as ports of $F'$.
      Note that there are exactly two satisfying assignments;
      all copies of $F$ must flip state together.
      Now assign each port of $F'$ to a dangling edge, respecting cyclic order.
      If there are extra dangling edges, put free terminals on them.

      The result is a generalized fanout with the same ports as $F'$,
      but the partition into the two legal configurations may be wrong.
      For each port on the wrong side of the partition,
      attach a NOT gate
      (composing simulations using Lemma~\ref{lem:compose sim}).
      This changes which of the two satisfying assignments
      has the edge at that port pointing in.
      Now the simulated gadget partition its ports
      into two legal configurations in the same way as $F'$,
      so it is in fact $F'$.
    \item This is a special case of the above.
  \end{enumerate}
\end{proof}

\begin{lemma}\label{lem:simnor}
  1-in-3 gates, free terminals, and fanouts
  parsimoniously simulate the NOR gate.
\end{lemma}

\begin{proof}
  The simulation consists of three
  1-in-3 gadgets,
  four free terminals,
  two 2-way fanouts,
  and a 3-way fanout,
  and is shown in Figure~\ref{fig:simnor}.
  The inputs are $a$ and $b$ and the output is $c$.
  True means `pointing right',
  for inputs, outputs, and labeled free terminals.

  The simulation enforces a 1-in-3 constraint on each of
  $\{a,x,c\}$, $\{b,y,c\}$, and $\{x,y,z\}$.
  The easiest way to verify that
  this is a parsimonious simulation of a NOR gate
  is to list all satisfying assignments:

  $$\begin{array}{cccccc}
    a & b & x & y & z & c \\\hline
    T & T & F & F & T & F \\
    T & F & F & T & F & F \\
    F & T & T & F & F & F \\
    F & F & F & F & T & T
  \end{array}$$

  We see that $c$ is always $\lnot(a\lor b)$,
  and there is a unique assignment for each combination
  of values for $a$ and $b$.
  When comparing against the definition of the NOR gate
  in Table~\ref{tbl:gadgets},
  recall that the correspondence between true/false
  and in/out
  is reversed for $c$ relative to $a$ and $b$.

  \begin{figure}
    \centering
    \begin{circuitikz}
      \draw
        (2.5,0) node[
        muxdemux,
        muxdemux def={NL=3, NB=0, Lh=3, Rh=2, w=0.4},
        rotate=180
        ] (out) {}

            (-1,1.33) node[
            muxdemux,
            muxdemux def={NL=2, NB=0, Lh=3, Rh=2, w=0.4},
            rotate=180
            ] (vtop) {}

            (-1,-1.33) node[
            muxdemux,
            muxdemux def={NL=2, NB=0, Lh=3, Rh=2, w=0.4},
            rotate=180
            ] (vbot) {}

            (vtop.lpin 2) -- (2.75,1.75)
            (vbot.lpin 1) -- (2.75,-1.75)

            (vtop.lpin 1) -| (0,0.25)
            (vbot.lpin 2) -| (0,-0.25)
            (0,0) circle (0.25) node {\sfrac13}

            (-2,1.33) node[left] {\(x\)} -- (vtop.rpin 1)
            (-2,0) node[left] {\(z\)} -- (-0.25,0)
            (-2,-1.33) node[left] {\(y\)} -- (vbot.rpin 1)

            (-3,3) node[left] {\(a\)} -- (3,3) -- (3,2) (3,1.75) circle (0.25) node {\sfrac13} (3,1.5) |- (out.lpin 3)
            (-3,-3) node[left] {\(b\)} -- (3,-3) -- (3,-2) (3,-1.75) circle (0.25) node {\sfrac13} (3,-1.5) |- (out.lpin 1)

            (out.lpin 2) -- ++(1,0) node[right] {\(c\)}
            (out.rpin 1) -- ++(-.67,0) node[left] {\(c\)}
            ;

    \end{circuitikz}
    \caption{
      A simulation of a NOR gate using
      1-in-3 gadgets,
      free terminals,
      and fanouts.
    }
    \label{fig:simnor}
\end{figure}
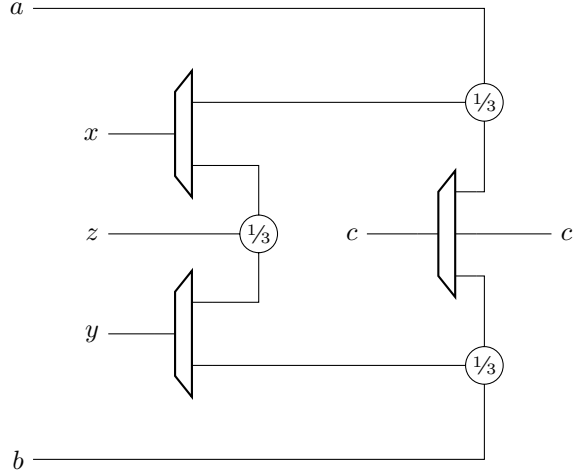

\end{proof}

\begin{lemma}\label{lem:simall}
NOR gates and fanout gates
parsimoniously simulate OR gates, AND gates, and crossovers.
\end{lemma}

\begin{proof}
In Section~\ref{sec:pgo gates} we observed that
boolean circuits can be converted to gadget networks;
these networks can be thought of as parsimonious simulations.
The NOR operation is logically complete
and can build crossovers in planar circuits \cite{monotone}.
Hence for each gate we want to simulate,
we can find a planar circuit of NOR gates that computes it,
then convert this into a planar network of NOR-gate gadgets and fanout-gate gadgets.
\end{proof}

We now pull everything together
by applying the results on simulation in Section~\ref{sec:pgo simulation}
to the hardness results of Section~\ref{sec:pgo hardness},
using the simulations in this section.
We can either compose reductions
or compose simulations using Lemma~\ref{lem:compose sim};
either way, we obtain our main result
about planar graph orientation decision problems.

\begin{corollary}\label{cor:simpler gadgets}
PGO consistency with
any generalized fanout,
fixed terminals,
free terminals,
and 1-in-3 gadgets
is NP-hard.
PGO promise-inference
and PGO uniqueness
with the same set of gadgets
are coNP-hard.
\end{corollary}

\subsection{Applying the framework}\label{sec:impl}

We now conclude the development of our gadget framework
by discussing what it takes to use it
to prove hardness for a game like Minesweeper.
We do not attempt to precisely define ``like Minesweeper'',
and this results of this section are not stated precisely.
Ultimately, one needs a polynomial-time reduction from the appropriate
PGO problem to the corresponding question for the game under consideration,
but the games we will consider are similar enough that
the reductions share most of their structure,
and the discussion here applies to many or all of them.

It is conceptually helpful to distinguish between
gadgets as abstract formal objects
and the constructions we build in concrete games
to embed gadgets in them.
We call the latter `implementations',
to parallel the distinction between
the specification and the implementation of a function.

\begin{definition}
An \defn{implementation} of a gadget is
a construction in a game like Minesweeper which behaves like the gadget,
in that there are covered cells (usually on the boundary of the construction)
corresponding to ports, and the construction has a solution
exactly when the configuration of edges
represented by the values of those cells
is legal for the gadget.
\end{definition}

Unfortunately, because it matches how the term is typically used,
we will frequently call implementations ``gadgets''
when the intended meaning is clear from context.

For implementations to interact in a way that simulates gadget networks,
we need some kind of \defn{wire}.
Each of our wires will have two possible states,
corresponding to the two orientations of the edge it represents.
We must be able to route our wires in an arbitrary planar graph---the
ability to extend and turn them is sufficient.
We need to be able to plug wires into the ports of implementations,
which sometimes requires the ability to adjust
the alignment of a wire by a single cell.

Our wires and implementations also need to have the following properties.
\begin{itemize}
  \item Every solution uses the same number of mines.
    This means the total number of mines in the puzzle doesn't
    reveal any information that might break our reduction.
  \item Every uncovered cell has a clue,
    as in most implementations of Minesweeper.
    There are versions of Minesweeper which
    allow cells that are known
    to be safe but don't have additional information
    (e.g. showing a question mark instead of a number),
    but we don't want to rely on this feature.
  \item All given mines can be deduced from uncovered non-mine cells.
    This provides robustness against changes in the definition,
    e.g. whether the locations of mines can be provided as part of a
    puzzle in addition to uncovered cells without mines.
\end{itemize}

Implementing some gadgets in a game
allows us to a build network of those gadgets in the game.
The resulting instance of the game has a solution
if and only if the network is satisfiable.

\begin{claim}
If a game like Minesweeper has implementations of the gadgets in $S$,
then there is a polynomial-time reduction
from PGO consistency with $S$
to the game's consistency problem.
\end{claim}

Note that our definition of Minesweeper consistency
in Section~\ref{sec:prior} is specific to Minesweeper.
For other games or variants, `partial board' and `consistent'
need to be defined appropriately,
but the other definitions (for inference and solvability) work as is.

For the other two decision problems,
we need well-behaved implementations.

\begin{definition}
An implementation is \defn{silent} if clicking
a known-safe internal cell can never reveal information
that wasn't already known.
In other words, for each covered internal cell,
the information that cell provides when clicked
(e.g. the number of adjacent mines)
is the same in every solution in which it doesn't have a mine.
\end{definition}

\begin{claim}
If a game like Minesweeper has silent implementations of the gadgets in $S$,
then there is a polynomial-time reduction
from PGO promise-inference with $S$
to the game's inference problem.
\end{claim}

\begin{proof}
Embed a network of gadgets
from an instance of PGO promise-inference with $S$ in the game.
For each (constant size) gadget in the network,
consider all solutions to its implementation
which are consistent with its semilocal constraint.
If there are \defn{commonalities},
or cells which are safe in all such solution
(or mines in all such solutions),
``click on'' them, revealing them in the instance of the game.
Silence guarantees that the information revealed
is the same in all such solutions,
so the information to put on that cell can be determined
from the semilocal constraint.
If a commonality dictates the orientation of an edge,
we click on all covered cells in the edge in the same way.
It doesn't matter whether this includes the cell(s)
representing the port on the other end of the edge---if
that gadget's semilocal constraint doesn't force the edge in the same way,
we must be in the first case of the promise
so there's an inferable locally unforced edge anyway.

If the network has an inferable locally unforced edge,
then cells inside the wire representing that edge
(or cells representing the ports that edge connects to)
can be inferred.

Otherwise, we are in the second case of the promise.
We've already clicked all cells that can be deduced
from each semilocal constraint---for any cell
that is still covered, the gadget (or wire) it's in
has solutions compatible with its semilocal constraint
in which that cell has a mine and in which it doesn't.
Thus the entire instance also has both of these kinds of solutions
because we are promised that every element of a semilocal constraint
is achieved by a satisfying assignment.
That is, the cell in question can't be inferred.
This is where it is crucial that we use promise-inference,
and not a simpler PGO inference decision problem.
\end{proof}

For PGO uniqueness and Minesweeper solvability,
we also need parsimony as we did for Lemma~\ref{lem:pgo sim uniqueness}.

\begin{definition}
An implementation is \defn{parsimonious} if it has exactly one solution
corresponding to each legal configuration of the gadget.
\end{definition}

\begin{claim}
If a game like Minesweeper has
silent parsimonious implementations of the gadgets in $S$,
then there is a polynomial-time reduction
from PGO uniqueness with $S$
to the game's solvability problem.
\end{claim}

\begin{proof}
Embed a network $N$ of gadgets
from an instance of PGO uniqueness with $S$ in the game.
Thanks to parsimony,
satisfying assignments of the network
are in bijection with consistent solutions to this instance of the game.
The secret solution is the solution corresponding to
the given satisfying assignment of $N$,
but cells are only uncovered if they are uncovered in the embedding of $N$,
which doesn't depend on the satisfying assignment.

If $N$ has another satisfying assignment,
the game instance has multiple solutions.
The player may be able to safely click some cells---perhaps
the orientation of some edge is forced,
or they can deduce that a cell inside a gadget is safe.
However, since our implementations are silent,
the player can't gain any information by doing this.
In particular, all solutions consistent with the initial state of the game
are also consistent after the player makes any sequence of safe clicks.
In order to solve the instance, they would need to distinguish
between these consistent solutions, which is impossible.

Conversely, suppose the only satisfying assignment to $N$
is the one given to the reduction.
Then the secret solution is the only solution consistent
with the initial state of the game
(we need parsimony to ensure there is only one).
Thus the value of every cell can be deduced,
and all cells without mines can be safely clicked in any order.
\end{proof}

It follows that if we can silently and parsimoniously
implement that gadgets needed by Corollary~\ref{cor:simpler gadgets},
we have hardness of consistency, inference, and solvability
for the game in question.
To simplify a little further,
fixed terminals are trivial to construct:
just let a wire end, and reveal some cells in it
to force its orientation.
Alternatively, modify a free terminal
(possibly including a bit of wire)
by revealing cells that determine its configuration.

\begin{corollary}\label{cor:simple}
  Suppose that for some game like Minesweeper,
  we have silent parsimonious implementations of
  any generalized fanout,
  free terminals,
  and 1-in-3 gadgets,
  and we are able to route, turn, and adjust the alignment of
  wires enough to embed any planar network of those gadgets.
  Then consistency is NP-hard,
  and inference and solvability are coNP-hard.
\end{corollary}

We now go through one final layer of abstraction,
which will handle routing of wires and filling empty space.
We lay a network of the gadgets above on a square grid,
where edges run horizontally and vertically and can turn.
In particular, each tile contains either one of the gadgets above,
a straight edge section, and turning edge section,
or nothing.
Tiles have up to one port on each edge, usually at the center.

This lets us reduce from problems about
planar networks on grids,
so for wire routing we need only construct tiles with straight and turning wires.

It turns out we don't need turning wires,
and can instead turn using any generalized fanout
and the other tiles.
As a first attempt, pick two ports on adjacent sides of the generalized fanout
and cover its other ports with free terminals.
This sometimes makes a turning wire,
but if the chosen ports have the same `polarity',
meaning they point in or out together,
it instead makes a turning NOT gate.
If we can only build turning NOT gates this way,
all ports of the generalized fanout must have the same polarity,
so its legal configurations are all-in and all-out.
But the generalized fanout must have two ports on opposite sides,
so we can make a straight NOT gate by covering the other ports with free terminals.
A turning NOT gate and a straight NOT gate in series gives a turning wire.

Finally,
for a few of our applications, instead of building a 1-in-3 gadget
we build a 2-in-3, 1-in-4, or 3-in-4.
With NOT gates and free terminals,
each of these can easily simulate a 1-in-3.

\begin{corollary}\label{cor:grid}
  Suppose that for some game like Minesweeper,
  we have silent parsimonious implementations of gadgets
  which are all square tiles that fit in a grid
  and interact appropriately with adjacent tiles.
  Suppose in particular that we have
  empty tiles, straight wires,
  any generalized fanout,
  free terminals,
  and any one of 1-in-3, 2-in-3, 1-in-4, or 3-in-4 gadgets.
  Then consistency is NP-hard,
  and inference and solvability are coNP-hard.
\end{corollary}

\section{Hardness proofs} \label{sec:variants}

To demonstrate the utility of this PGO framework
and the ease with which it facilitates writing hardness proofs,
we provide a number of example applications
that prove hardness of Minesweeper and Minesweeper-like games.
In particular,
our examples are drawn from the video game \emph{14 Minesweeper Variants},
which despite its name actually implements significantly more than 14 Minesweeper variants.

\renewcommand\v[1]{\textsf{[\makebox[1.1em][c]{#1}]}}
\newcommand\vn[2]{\v{#1} \textbf{#2}:}

The variants in the game are identified by letters surrounded with brackets.
For instance,
\v V stands for \emph{vanilla} Minesweeper with no variants.
We first briefly summarize the rules of each variant in the game.

\medskip
The rules are separated into two categories.
The first category,
which we refer to as ``left rules''
(because they appear on the left half of the menu),
contains rules that add global constraints to the board
without changing the meanings of clues:

\medskip
\begin{tabularx}{\dimexpr\linewidth-\parindent}{lX}
    \vn{Q}{Quad}        & there must be at least 1 mine in every \(2\times2\) region \\
    \vn{C}{Connected}   & all mines are orthogonally or diagonally connected \\
    \vn{T}{Triplet}     & there can be no 3 mines in a row (orthogonally or diagonally) \\
    \vn{O}{Outside}     & all non-mines are orthogonally connected; all mines are orthogonally connected to any edge of the grid \\
    \vn{D}{Dual}        & all mines form non-orthogonally-touching \(1\times2\) dominoes \\
    \vn{S}{Snake}       & all mines form an orthogonally connected path which does not touch itself orthogonally \\
    \vn{B}{Balance}     & the number of mines in each row and column is the same \\
    \vn{T'}{Triplet'}   & mines must be part of a 3-in-a-row orthogonally or diagonally \\
    \vn{D'}{Battleship} & all mines form non-touching (even diagonally) \(1\times1\)'s, \(1\times2\)'s, \(1\times3\)'s, or \(1\times4\)'s \\
    \vn{A}{Anti-knight} & no two mines can be a knight's move away from each other \\
    \vn{H}{Horizontal}  & no two mines can touch horizontally \\
    \vn{U}{Unary}       & no two mines can touch orthogonally
\end{tabularx}
\medskip

The second category,
which we refer to as ``right rules,''
contains variants that apply to the clues themselves
(which by default count the number of mines in the 8 orthogonally or diagonally adjacent cells):

\medskip
\begin{tabularx}{\dimexpr\linewidth-\parindent}{lX}
    \vn{M}{Multiple}      & the grid is checkerboard colored; each mine in a shaded cell counts as 2 \\
    \vn{L}{Liar}          & each clue is off by exactly 1 \\
    \vn{W}{Wall}          & clues give the multiset of lengths of contiguous runs of mines in the 8 orthogonally or diagonally adjacent cells \\
    \vn{N}{Negation}      & the grid is checkerboard colored; clues give the difference between the number of mines in shaded and unshaded cells \\
    \vn{X}{Cross}         & clues give the number of mines up to 2 away in all 4 orthogonal directions \\
    \vn{P}{Partition}     & clues give the number of groups of mines (i.e.\ number of [W] clues there would be) \\
    \vn{E}{Eyesight}      & clues give the number of non-mines seen in all 4 orthogonal directions (including the clue itself), where mines block line of sight \\
    \vn{X'}{Mini-cross}   & clues give the number of mines 1 away in all 4 orthogonal directions \\
    \vn{K}{Knight}        & clues give the number of mines a knight's move away \\
    \vn{W'}{Longest wall} & clues give the largest length of contiguous mines (i.e. the largest [W] clue there would be) \\
    \vn{E'}{Eyesight'}    & clues give the difference of non-mines seen horizontally and vertically, where mines block line of sight
\end{tabularx}
\medskip

In most cases,
our reductions will ``incidentally''
satisfy some number of rules,
in that the reduction is unaffected
or can be easily adapted to work
whether the rule is included or not.
This is very efficient,
as a reduction that incidentally satisfies \(n\) extra rules
constitutes a hardness proof for \(2^n\) different variants.
\medskip

% To define the decision problems for these Minesweeper variants,
% the definition of \defn{consistent} changes,
% and in some cases (e.g. \vn{W}) we need to redefine
% \defn{partial board} because uncovered cells contain a different kind of data.
% After making these changes, the definitions of the
% inference and solvability problems for the variant are the same as before.

The notation used in the below gadgets is as follows:
\begin{itemize}
    \item
        Yellow-highlighted cells represent wires.
        It is enforced by the given clues
        that e.g.\ for the wire \textsf A,
        either all cells marked \textsf A are mines
        and all cells marked \(\overline{\mathsf A}\) are empty,
        or all cells marked \textsf A are empty
        and all cells marked \(\overline{\mathsf A}\) are mines.
    \item
        Black cells represent mines
        which can be locally derived from the clues contained within the same gadget
        (or in some cases the connecting regions between gadgets).
        Usually,
        it is immediately obvious how to derive these
        (they come directly from a single clue);
        mines that require deductions that are any more complicated than that
        are marked with red stars.
    \item
        White cells represent pre-revealed cells,
        which contain the appropriate type of clue.
        In most cases,
        the appropriate number is clear from the surrounding cells;
        a few numbers are given explicitly for clarity.
\end{itemize}

When two of our gadgets meet, there will be a pair of mines,
with one in each gadget,
which interact in that exactly one of them must be a mine.
We think of the edge between them as pointing towards the mine.
For most versions, tiles are glued together as depicted in Figure~\ref{fig:glue}.

All of our gadgets satisfy the properties listed in Section~\ref{sec:impl},
so we have hardness even when the number of mines is known,
when all uncovered cells have clues,
and when the input can't contain given mines.

\begin{figure}
    \centering
    \fig{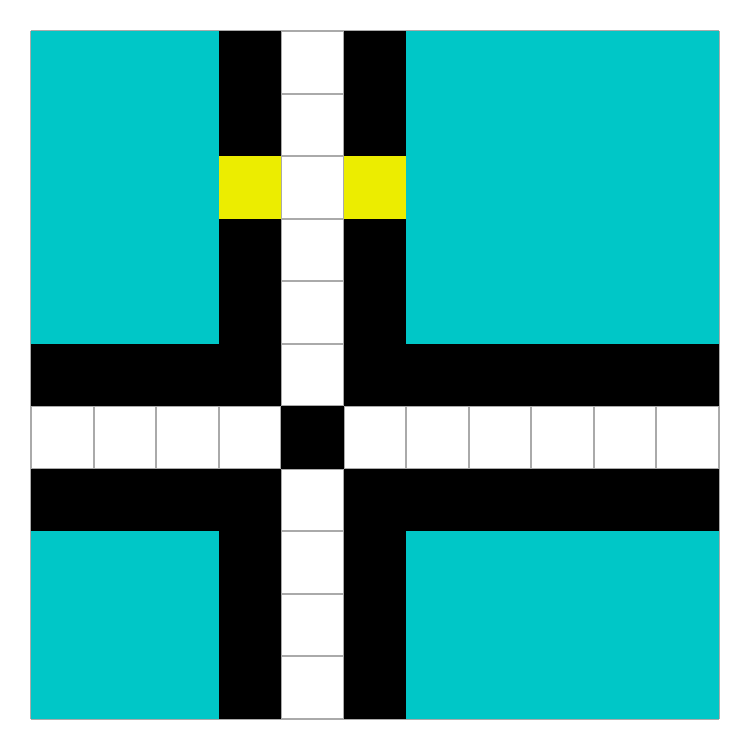}
    \caption{%
        Relative positioning of tile-based gadgets.
        Cyan represents the interior of each gadget,
        which differs between gadgets,
        and everything else shows the common ``frame.''
        The extra mine outside the frames
        is placed to satisfy \v C constraints.
        Details differ between variants (e.g.\ dimensions of each tile),
        but the structure is the same.%
    }
    \label{fig:glue}
\end{figure}

\subsection{Vanilla clues}

These gadgets work with vanilla clues together with any combination
of \v M, \v L, \v Q, \v C, and \v{T'}.

\medskip\noindent
\begin{center}
  \begin{tabular}{ccc}
      free terminal & straight wire & empty space \\
      \fig{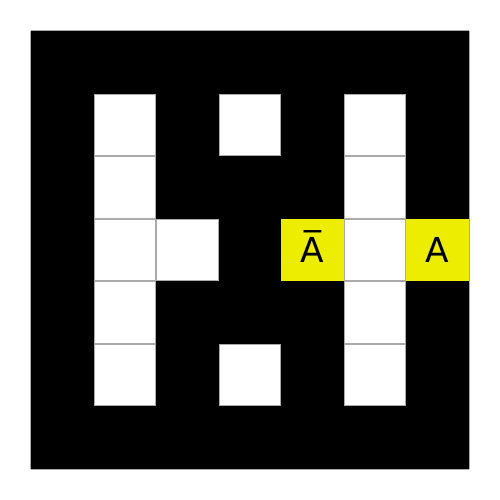} & \fig{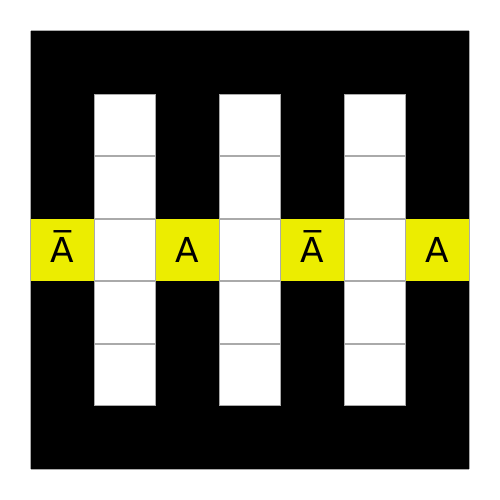} & \fig{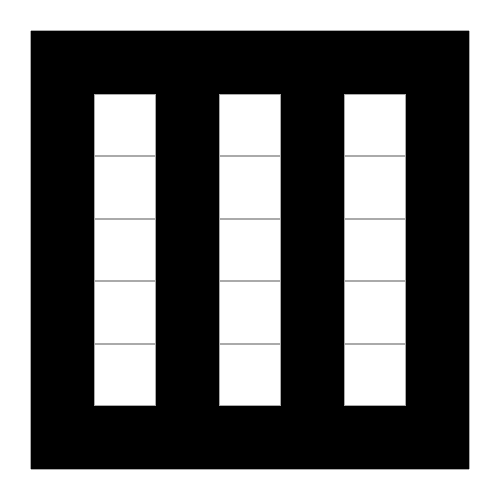}
  \end{tabular}
  \begin{tabular}{cc}
      fanout & 1-in-3 gadget \\
      \fig{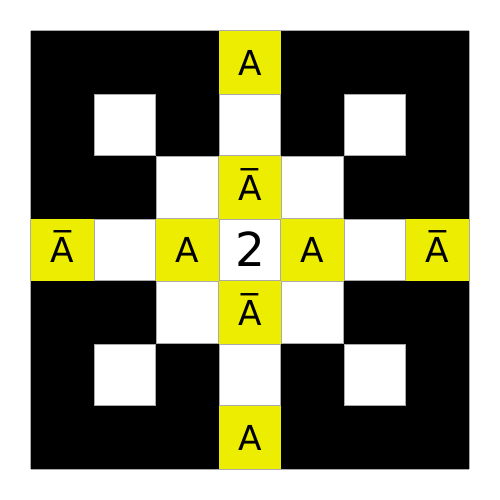} & \fig{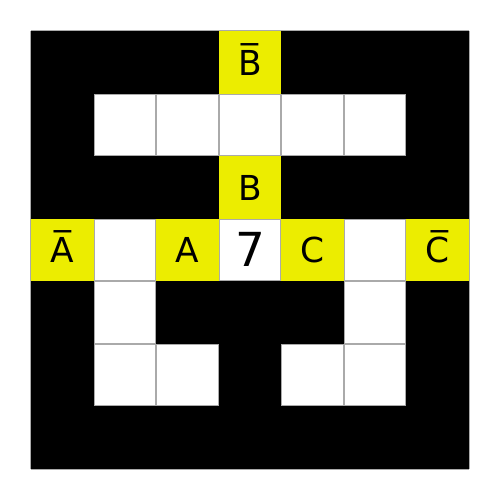}
  \end{tabular}
\end{center}

\subsection{Solvability from an empty board}

In the original Minesweeper game,
the player does not start with any cells revealed,
and the player starts by clicking any cell
on the entirely unrevealed board
(which the game guarantees to be safe).
The above gadgets do not work in this setting,
because they rely on a large number of clues being pre-revealed.
However,
we also construct the following ``transparent'' gadgets,
where all clues can be deduced starting from this initial state.
This shows that Minesweeper solvability is coNP-complete
even in the setting where the player is initially given no clues.
(To guarantee the click is safe,
we can double the width of the grid,
and delay the decision of which half to place the circuit on
until after the click is made.)

\medskip\noindent
\begin{center}
  \begin{tabular}{ccc}
      free terminal & straight wire & turning wire \\
      \fig[valign=m]{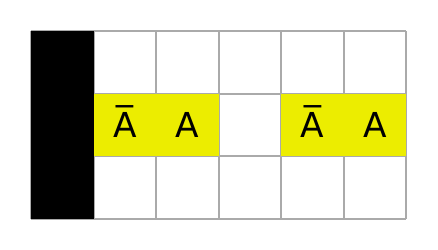} & \fig[valign=m]{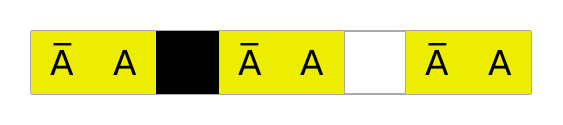} & \fig[valign=m]{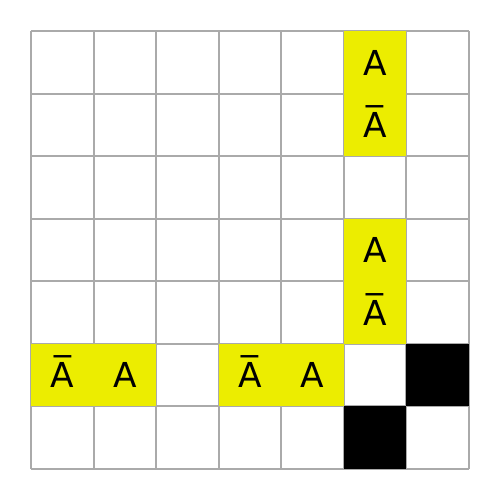} \\
      fanout & 1-in-3 gadget & shift \\
      \fig[valign=m]{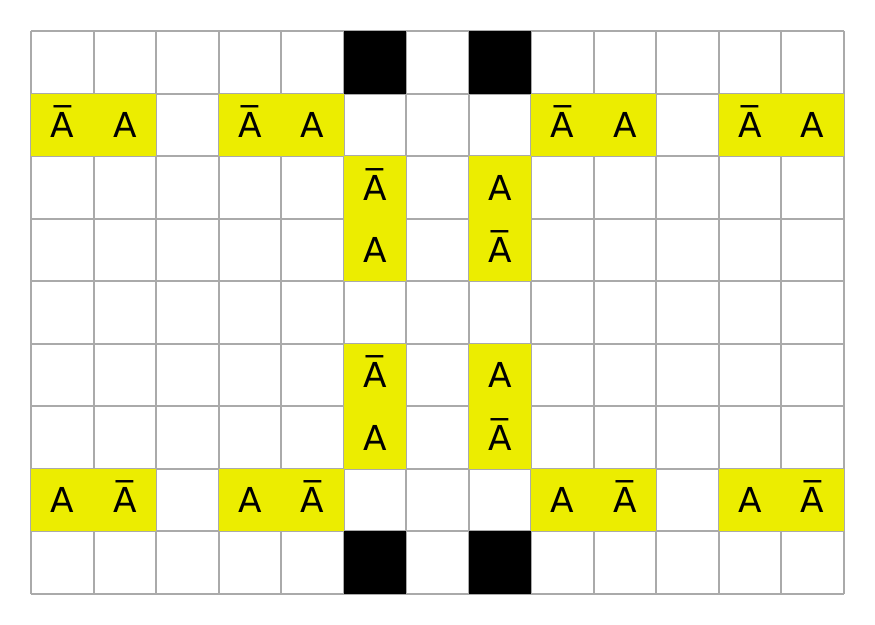} & \fig[valign=m]{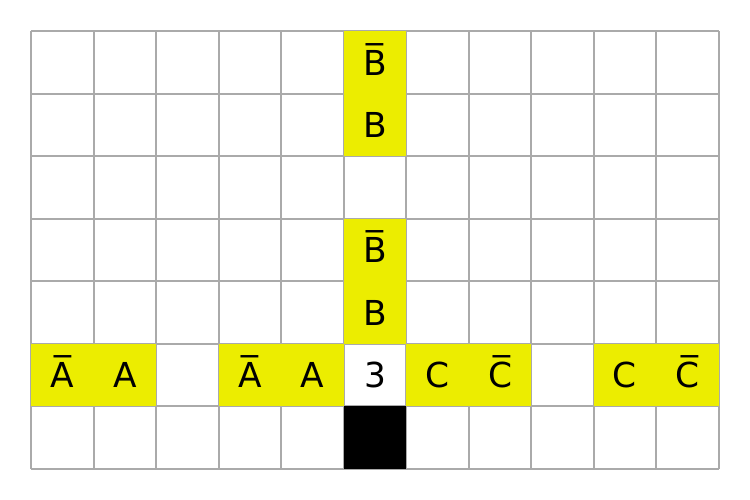} & \fig[valign=m]{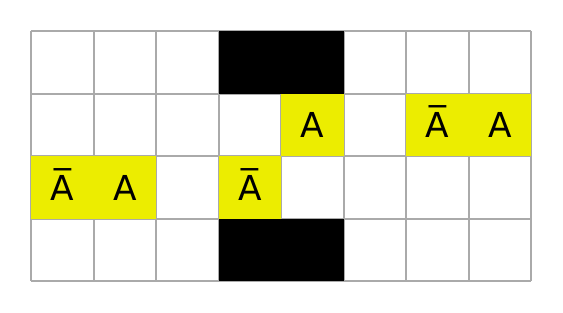}
  \end{tabular}
\end{center}

\subsection{Cross clues}

These gadgets work with \v X together with any combination
of \v M, \v L, \v Q, \v C, and \v{T'}.

\medskip\noindent
\begin{center}
  \begin{tabular}{ccc}
      free terminal & straight wire & empty space \\
      \fig[scale=0.87]{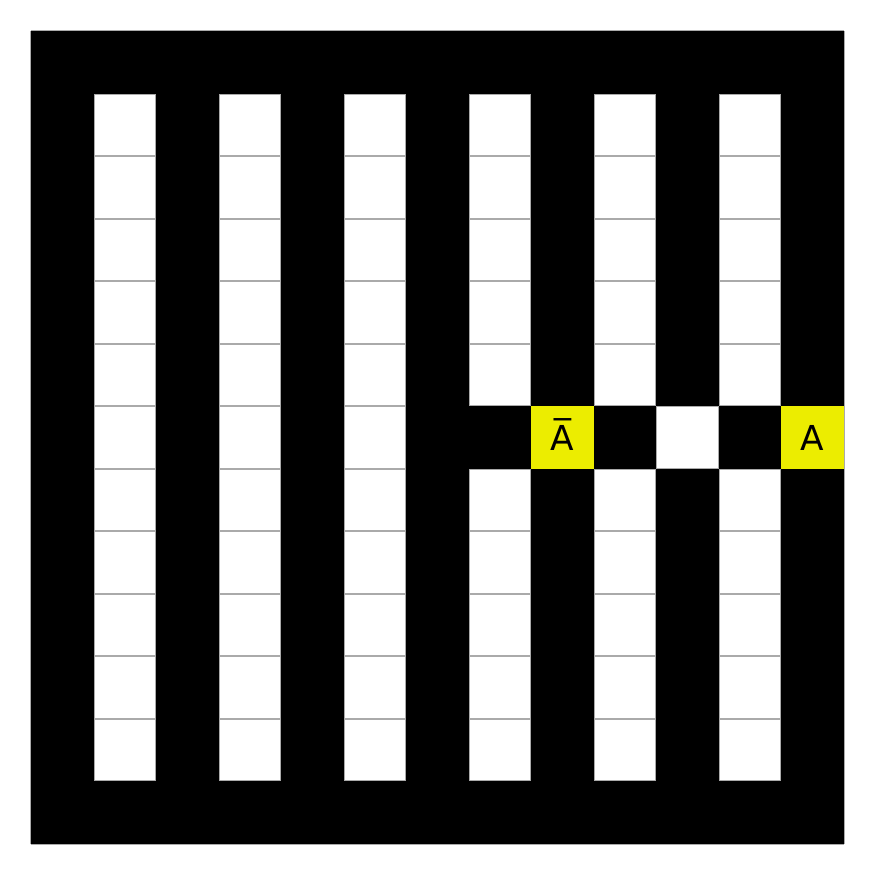} & \fig[scale=0.87]{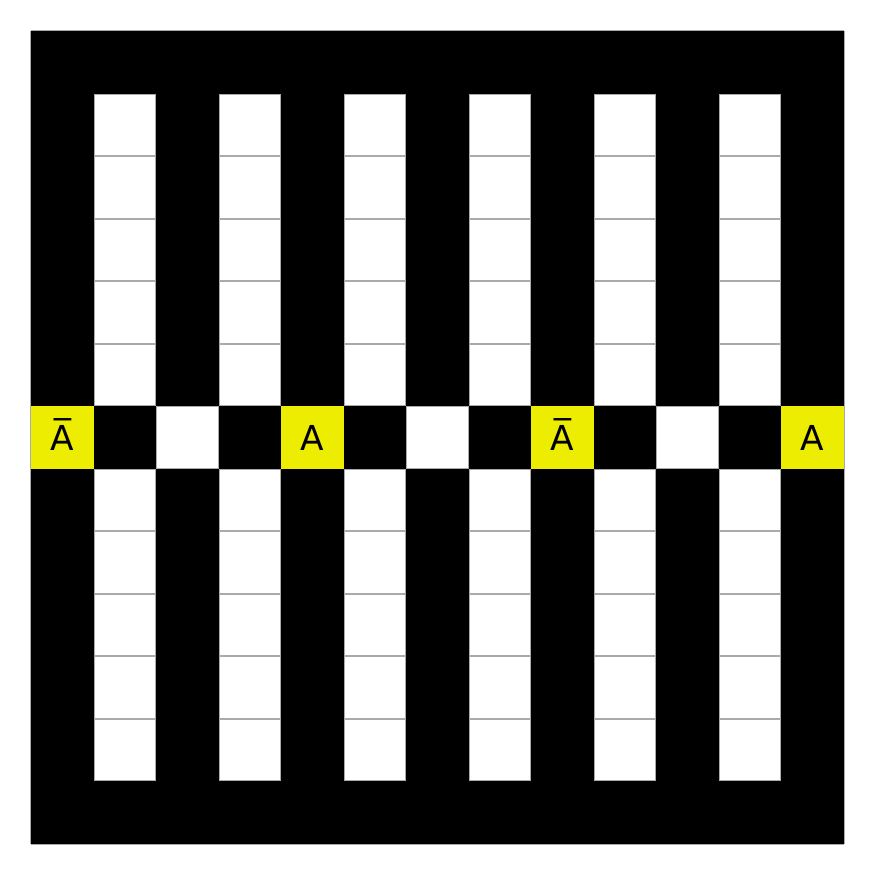} & \fig[scale=0.87]{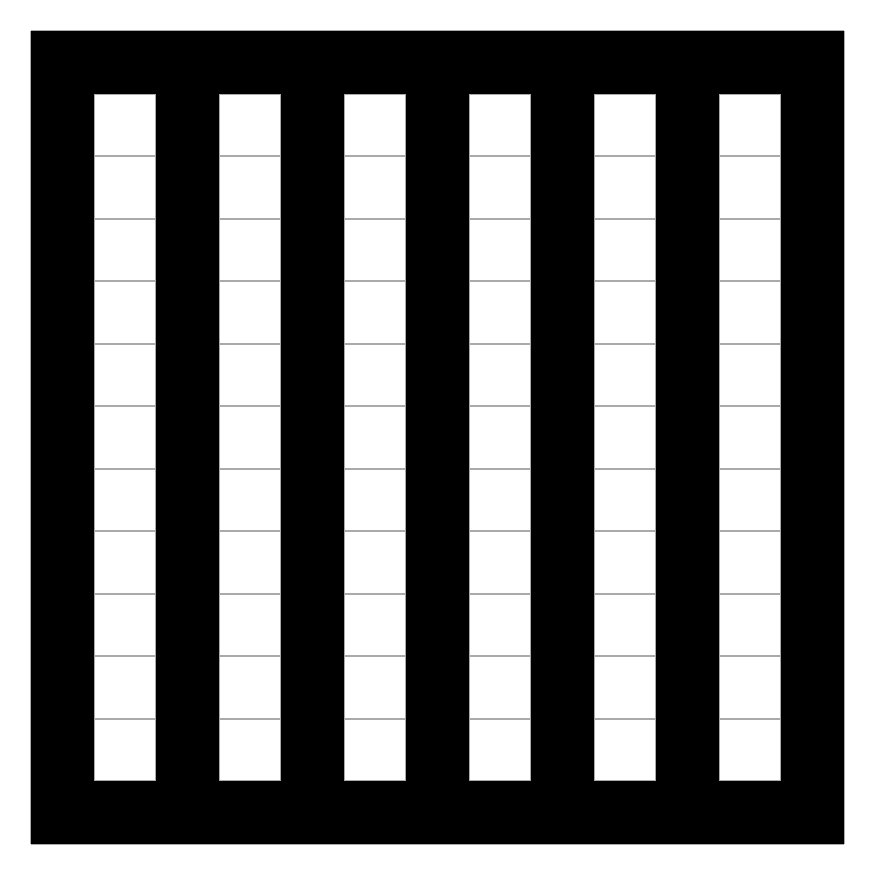}
  \end{tabular}
  \begin{tabular}{cc}
      fanout & 3-in-4 gadget \\
      \fig[scale=0.87]{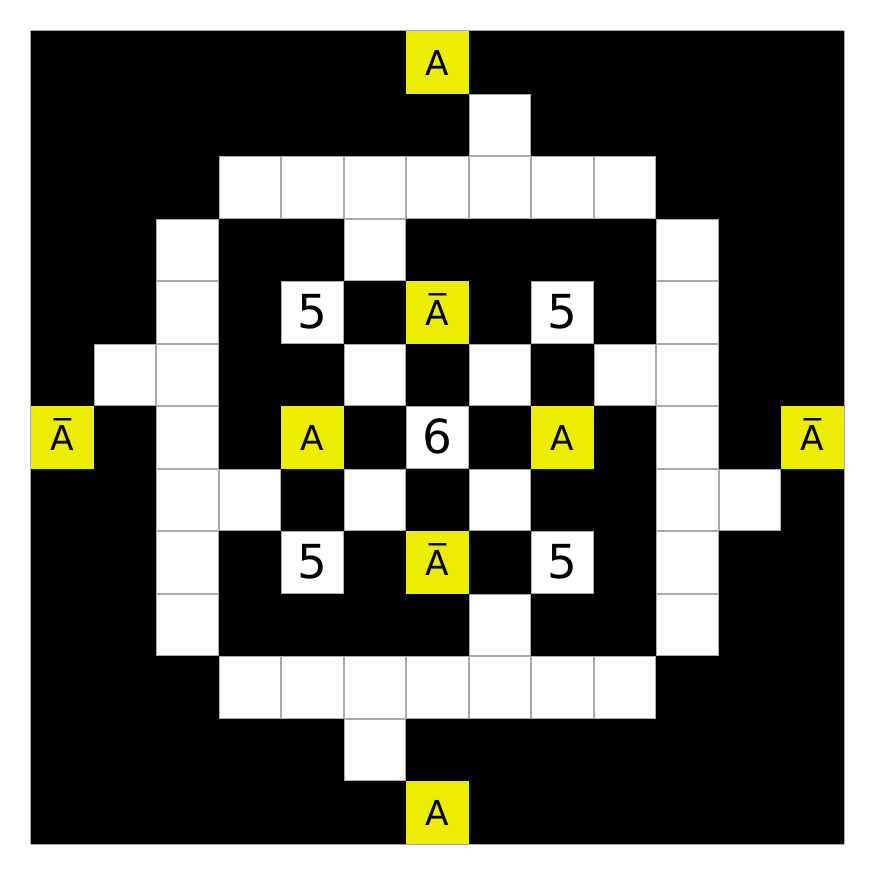} & \fig[scale=0.87]{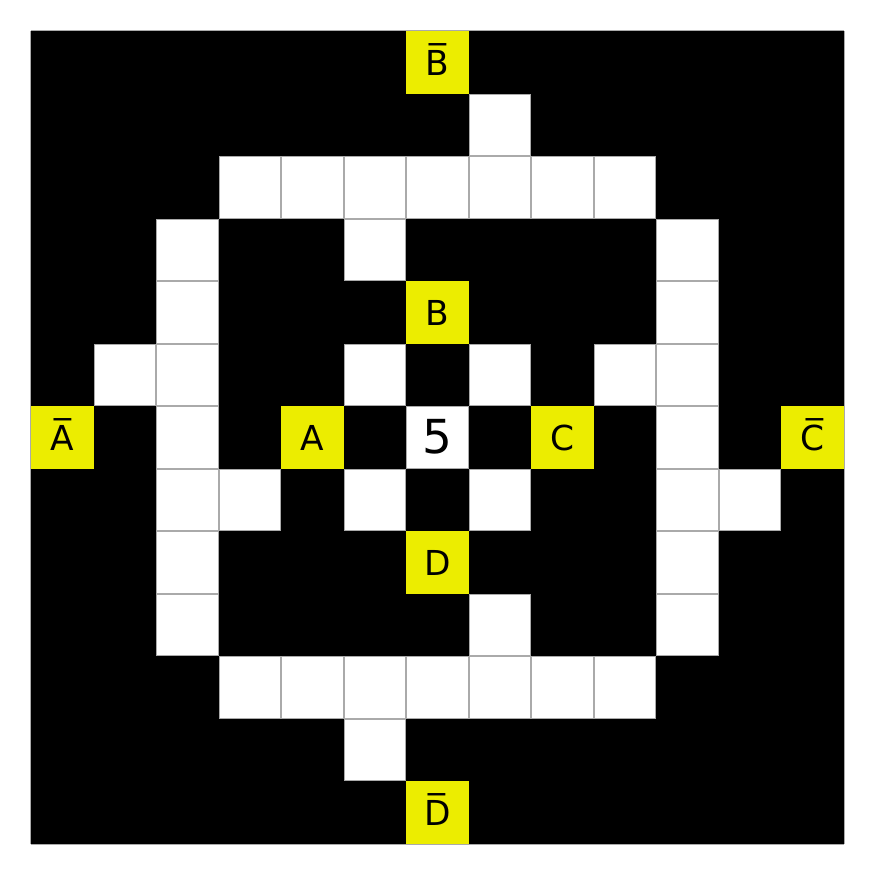}
  \end{tabular}
\end{center}

\subsection{Mini-cross clues}

These gadgets work with \v{X'} together with any combination
of \v M, \v L, \v Q, \v C, \v{T'}, and \v N.

\medskip\noindent
\begin{center}
  \begin{tabular}{ccc}
      free terminal & straight wire & empty space \\
      \fig{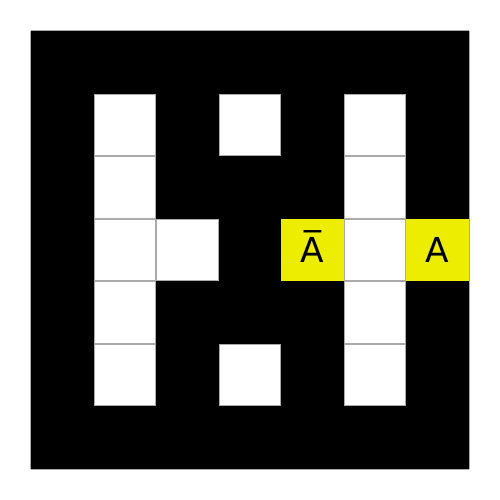} & \fig{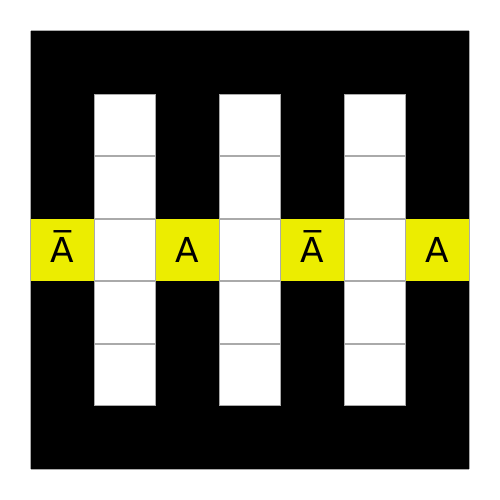} & \fig{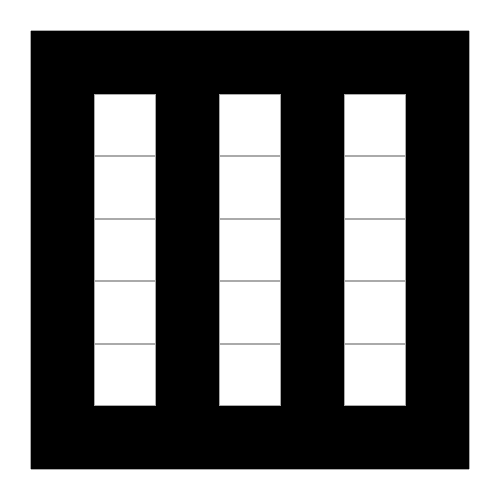}
  \end{tabular}
  \begin{tabular}{cc}
      fanout & 3-in-4 gadget \\
      \fig{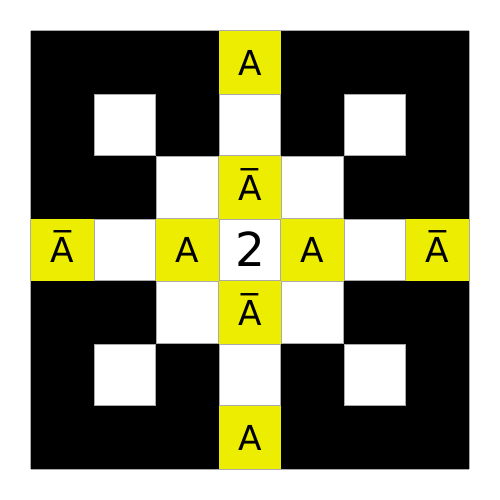} & \fig{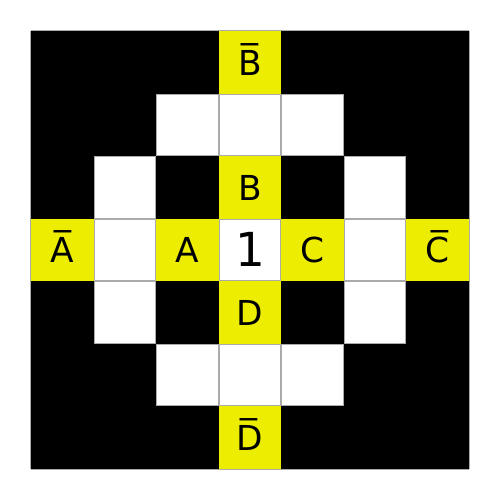}
  \end{tabular}
\end{center}

\subsection{Eyesight clues}

These gadgets work with \v E together with any combination
of \v L, \v Q, \v C, and \v{T'}.
Alternating tiles are reflected to make ports align.
Solutions don't all have the same number of mines in a single tile,
but they do globally:
the two cells interacting between tiles contain exactly one mine,
and ignoring those cells each tile has a constant number of mines.

\medskip\noindent
\begin{center}
  \begin{tabular}{ccc}
      free terminal & straight wire & empty space \\
      \fig[scale=0.87,valign=m]{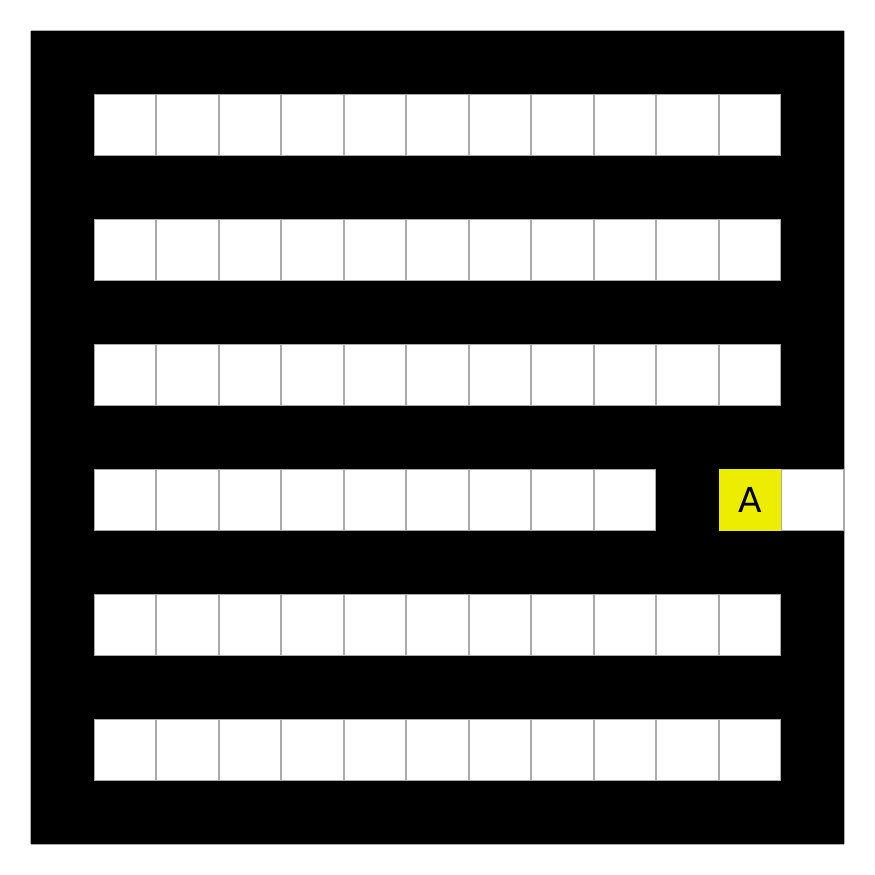} & \fig[scale=0.87,valign=m]{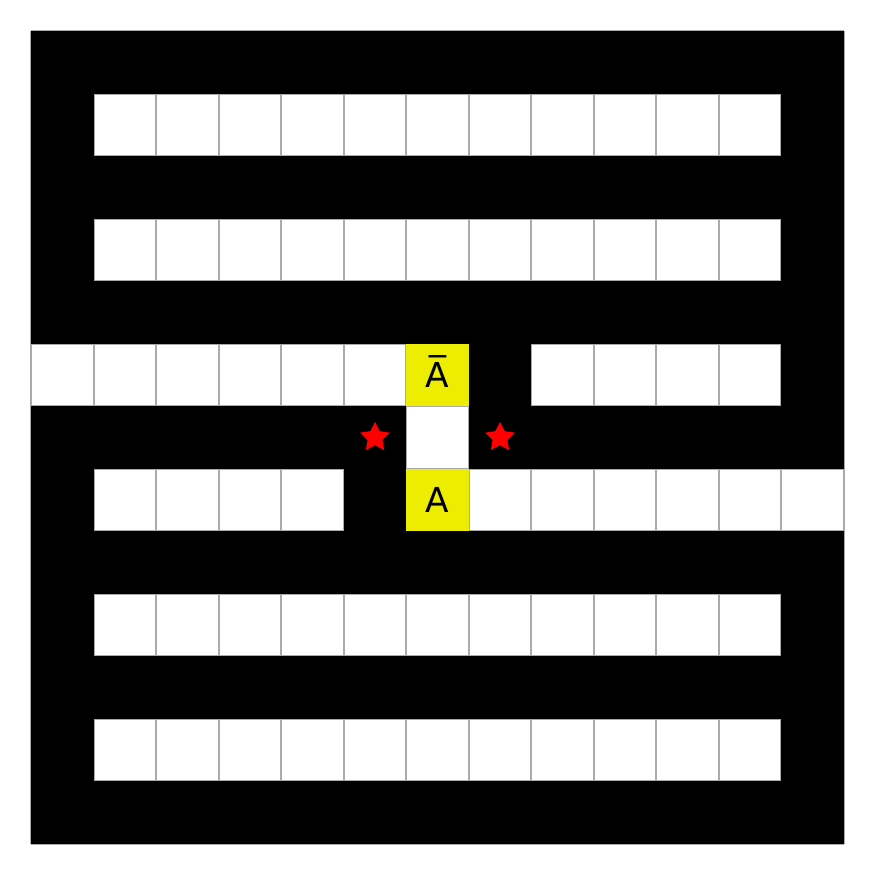} & \fig[scale=0.87,valign=m]{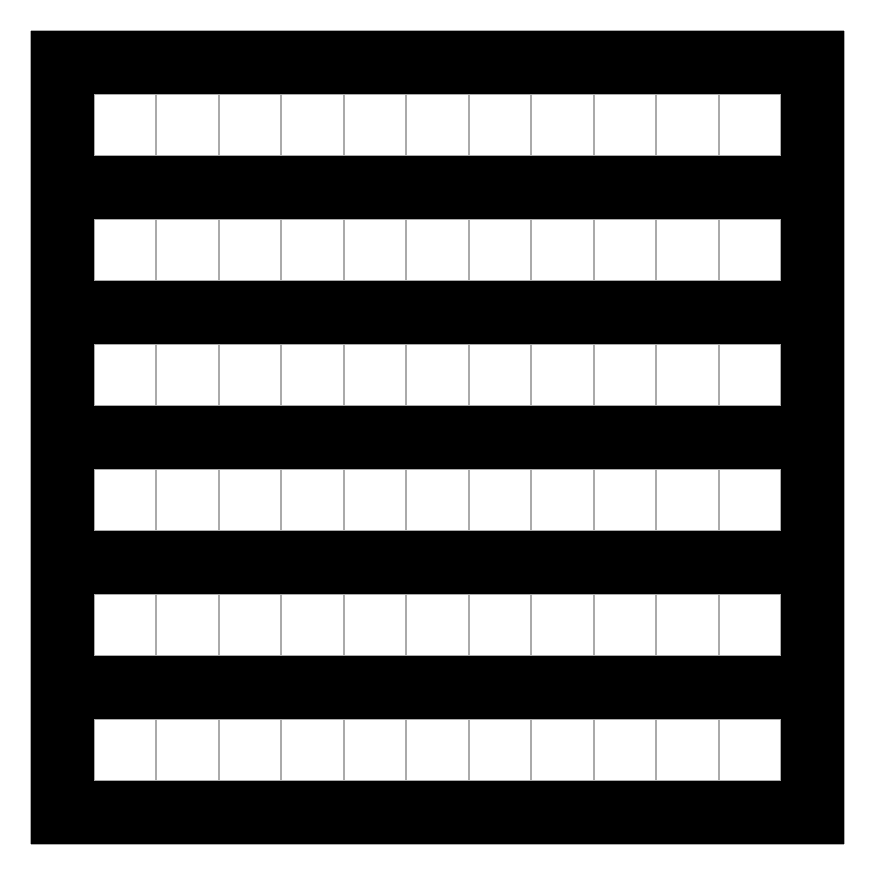}
  \end{tabular}
  \begin{tabular}{cc}
      fanout & 3-in-4 gadget \\
      \fig[scale=0.87]{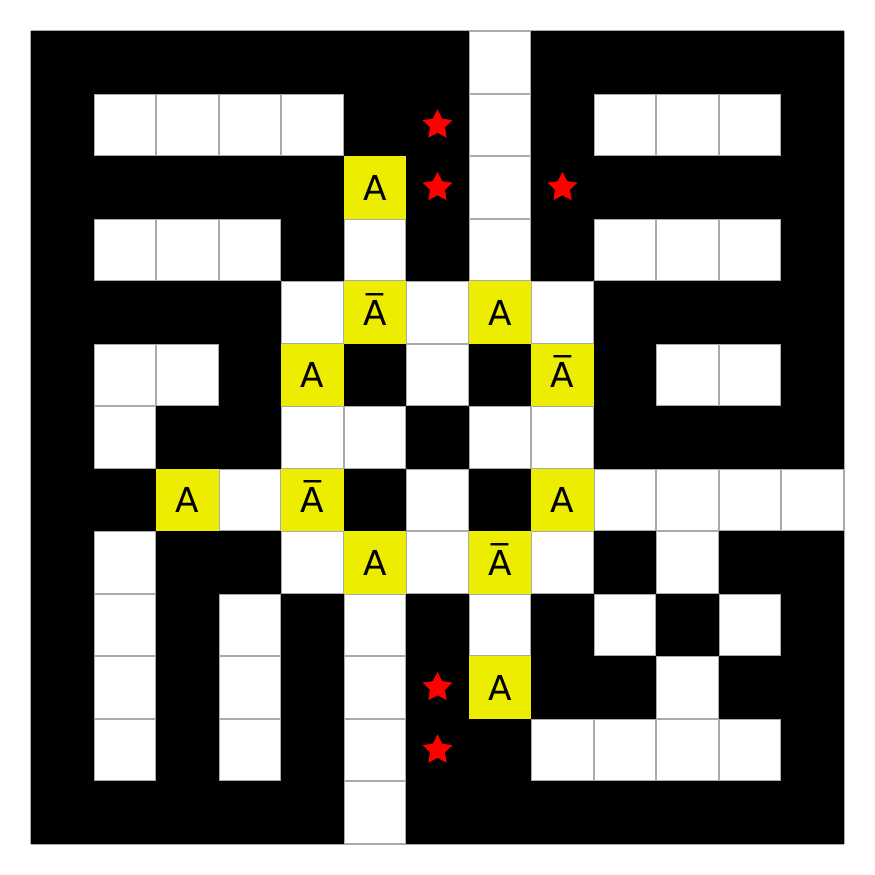} & \fig[scale=0.87]{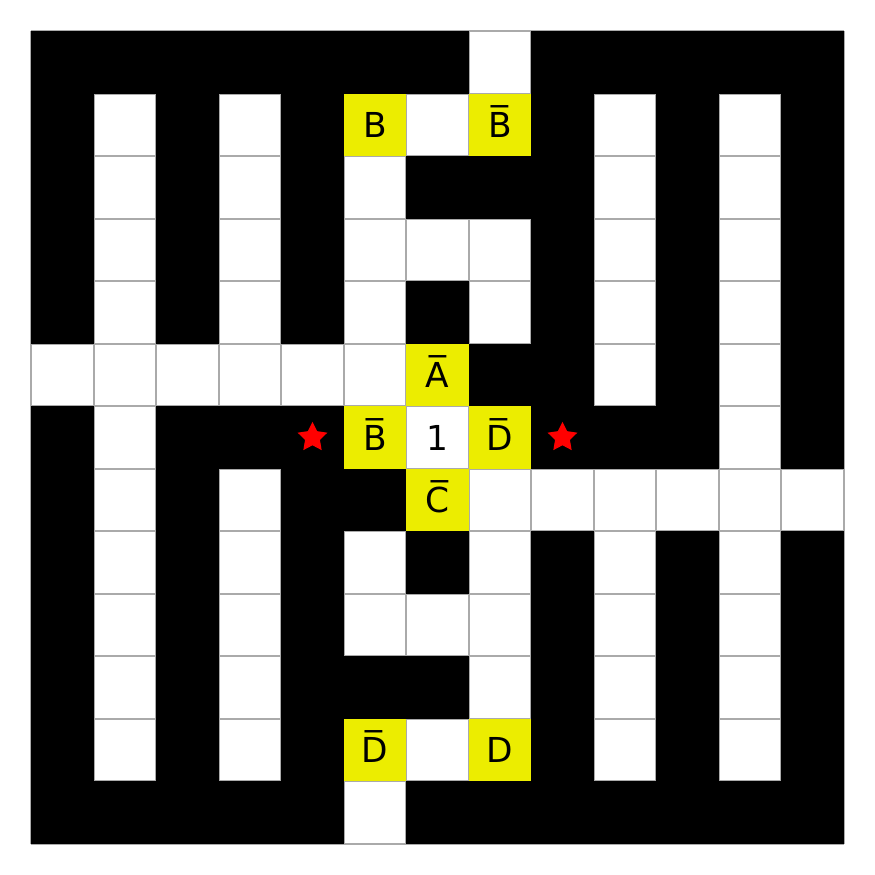}
  \end{tabular}
\end{center}

% \newpage % TODO temporary measures
\subsection{Wall and partition clues}

These gadgets work with either \v W or \v P, together with any combination
of \v U, \v T, and \v L.

\medskip\noindent
\begin{center}
  \begin{tabular}{cccc}
      \multicolumn{2}{c}{free terminal} & turning wire \\
      \multicolumn{2}{c}{\fig[valign=m]{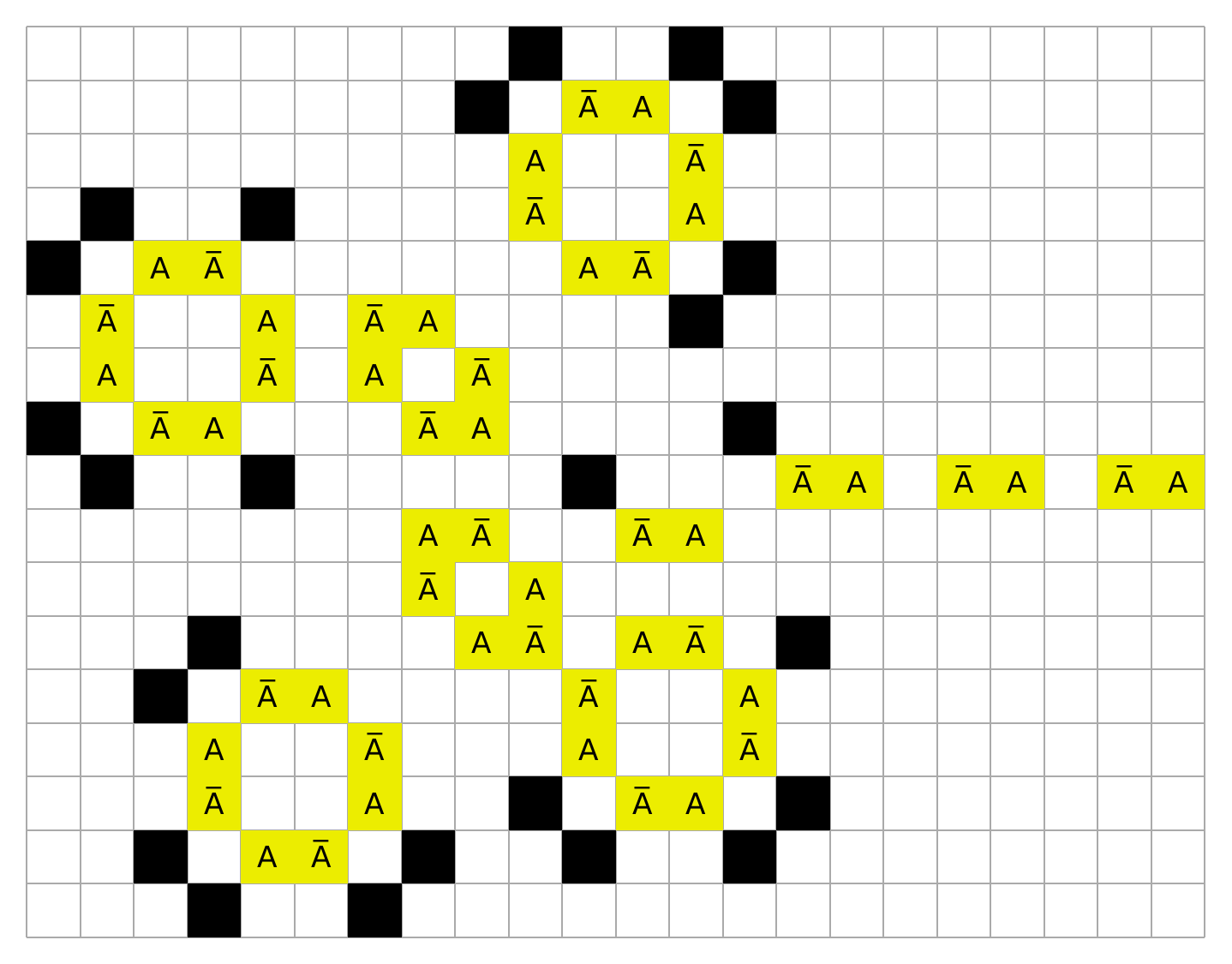}} & \fig[valign=m]{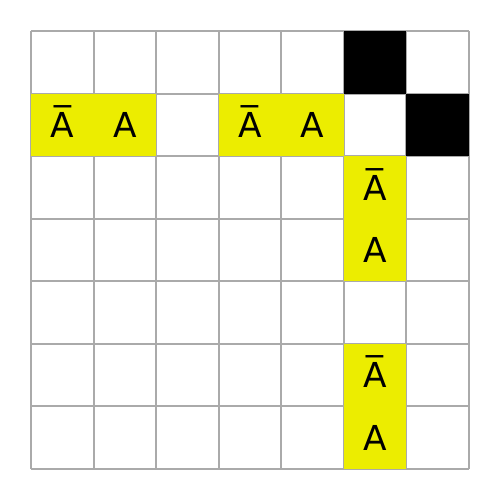} \\
      fanout & 1-in-3 gadget & shift \\
      \fig[valign=m]{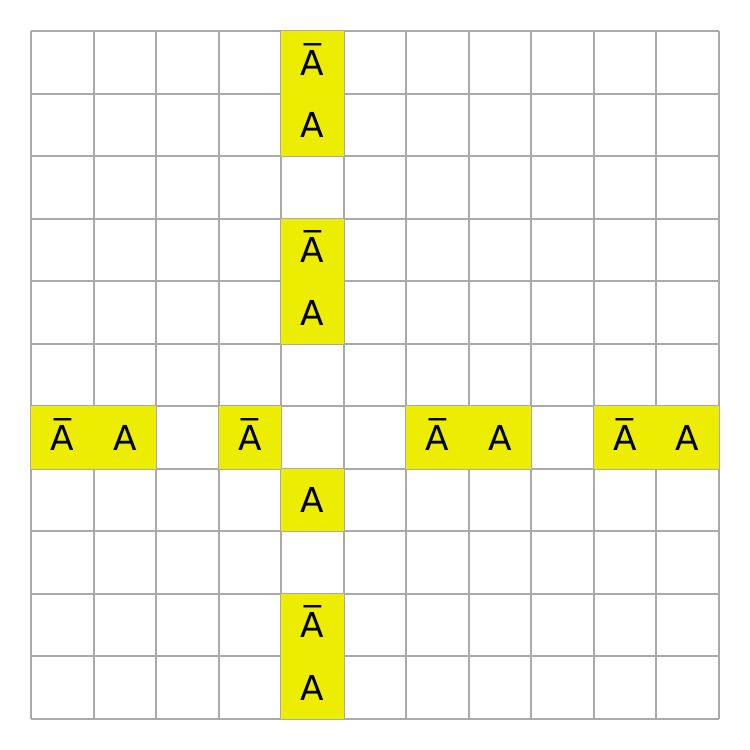} & \fig[valign=m]{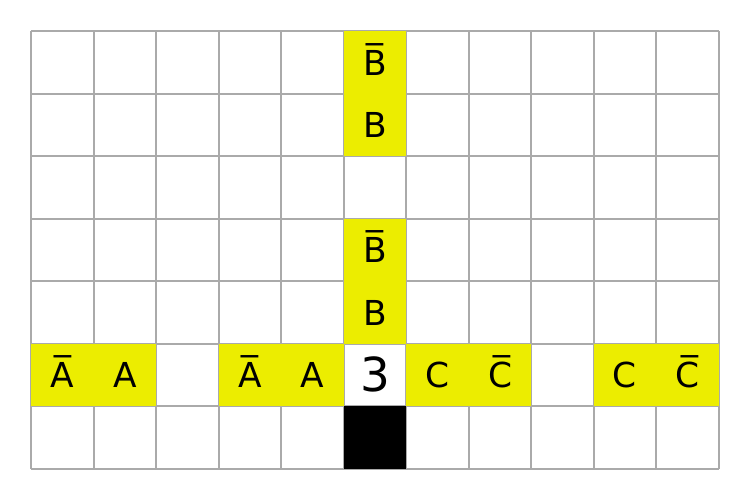} & \fig[valign=m]{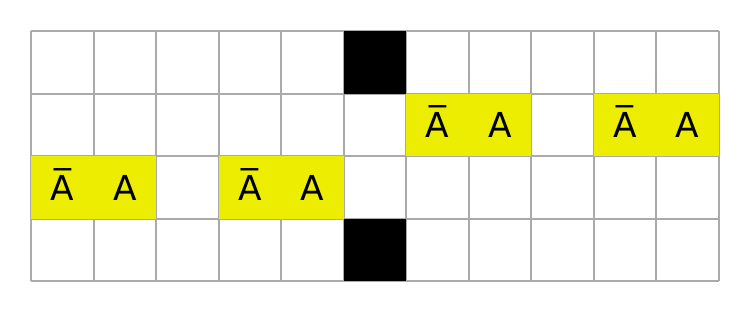}
  \end{tabular}
\end{center}

% \subsection{Negation clues}

% (TODO: won't bother putting this in until i get more variants)

% \section{Summary of results} \label{sec:summary}

% (TODO: this section will have a very pretty table eventually)

\printbibliography

\end{document}